\documentclass[11pt]{article}
\usepackage[table]{xcolor}
\definecolor{light-gray}{gray}{0.90}
\usepackage{amsmath}
\usepackage{amsthm}
\usepackage{enumerate}
\usepackage{graphicx,psfrag,epsf}
\usepackage{amssymb}
\usepackage{hyperref}
\usepackage{diagbox}
\usepackage{geometry}
\usepackage[utf8]{inputenc}
\usepackage{siunitx}
\usepackage{float}
\usepackage{todonotes}
\usepackage[title]{appendix}
\usepackage{algorithm}
\usepackage{algpseudocode}
\usepackage[none]{hyphenat}
\usepackage{titlesec}
\usepackage[official]{eurosym}
\usepackage{mathtools}
\usepackage{natbib}
\usepackage{caption}
\usepackage{subfigure}
\usepackage{color}
\usepackage{algorithm}
\usepackage{algorithmicx}
\usepackage{algpseudocode}
\usepackage{color}
\usepackage{extarrows}
\usepackage{booktabs}
\usepackage{bm}
\usepackage{authblk}

\newcommand{\shtarkov}{\ensuremath{\textsc{shtarkov}}}
\newcommand{\mmreg}{\ensuremath{\textsc{mmreg}}}
\newcommand{\mreg}{\ensuremath{\textsc{mreg}}}
\newcommand{\mmred}{\ensuremath{\textsc{mmred}}}
\newcommand{\reg}{\ensuremath{\textsc{reg}}}

\setuptodonotes{fancyline, color=blue!30}

\numberwithin{equation}{section}

\newcolumntype{C}[1]{>{\centering\arraybackslash}m{#1}}
\newcolumntype{R}[1]{>{\raggedleft\arraybackslash}m{#1}}
\definecolor{mygray}{gray}{0.9}

\newtheorem{proposition}{Proposition}

\newtheorem{theorem}{Theorem}

\newtheorem{example}{Example}

\theoremstyle{plain}
\newtheorem{definition}{Definition}
\newtheorem{lemma}{Lemma}

\newcommand{\grow}{\ensuremath{\textsc{GROW}}}

\newcommand{\bd}{\ensuremath{\textsc{bd}}}
\newcommand{\nv}{\ensuremath{0}}

\newcommand{\Sgrow}{\ensuremath{{S}_{\textsc{grow}}}}
\newcommand{\Srel}{\ensuremath{{S}_{\textsc{rel}}}}

\newcommand{\sgn}{\ensuremath{\textsc{sgn}}}

\newcommand{\conv}{\ensuremath{\textsc{conv}}}

\renewcommand{\vec}[1]{\ensuremath{{\bm #1}}}

\newcommand{\cU}{\ensuremath{\mathcal{U}}}
\newcommand{\cE}{\ensuremath{\mathcal{E}}}
\newcommand{\cH}{\ensuremath{\mathcal{H}}}

\newcommand{\cP}{\ensuremath{\mathcal{P}}}
\newcommand{\cS}{\ensuremath{\mathcal{S}}}
\newcommand{\cY}{\ensuremath{\mathcal{Y}}}
\newcommand{\cX}{\ensuremath{\mathcal{X}}}
\newcommand{\cV}{\ensuremath{\mathcal{V}}}
\newcommand{\cR}{\ensuremath{\mathcal{R}}}

\newcommand{\reals}{\ensuremath{\mathbb{R}}}
\newcommand{\naturals}{\ensuremath{\mathbb{N}}}

\newcommand{\meanspace}{\ensuremath{\text{\tt M}}}

\newcommand{\Pmu}{\ensuremath{P}}

\newcommand{\pmv}{\ensuremath{\bar{p}}}
\newcommand{\Pmv}{\ensuremath{\bar{P}}}
\newcommand{\comp}{\ensuremath{\textsc{c}}}

\newcommand{\commentout}[1]{}

\title{Growth-Optimal E-Variables and an extension to the multivariate Csisz\'ar-Sanov-Chernoff Theorem}

\begin{document}
\author[12]{Peter Grünwald} \author[1]{Yunda Hao}
\author[3]{Akshay Balsubramani}
\affil[1]{Centrum Wiskunde \& Informatica, Amsterdam, The Netherlands}
\affil[2]{Leiden University, Leiden, The Netherlands}
\affil[3]{Sanofi, Waltham, USA}

\bibliographystyle{plainnat}
\maketitle
%
\begin{abstract}
 We consider growth-optimal e-variables with maximal e-power, both in an absolute and relative sense, for simple null hypotheses for a $d$-dimensional random vector, and multivariate composite alternatives represented as a set of $d$-dimensional means $\meanspace_1$. These include, among others, the set of all distributions with mean in $\meanspace_1$, and the exponential family generated by the null restricted to means in $\meanspace_1$. We show how these optimal e-variables are related to Csisz\'ar-Sanov-Chernoff bounds, first for the case 
 that $\meanspace_1$ is convex (these results are not new; we merely reformulate them)  and then for the case that $\meanspace_1$ `surrounds' the null hypothesis (these results are new). 
 
\end{abstract}

\section{Introduction}
$E$-variables present a compelling alternative  to traditional $P$-values, particularly in hypothesis tests involving optional stopping and continuation \citep{GrunwaldHK19, VovkW21, Shafer21, ramdas2023savi, henzi2021valid, Grunwald23}. 
As is well-known, there is a close connection between optimal rejection regions of anytime-valid tests at fixed level $\alpha$ and optimal anytime-valid concentration inequalities \citep{howard2021time}. 
In this paper we consider a variation of this connection in the context of a simple multivariate null and several types of composite alternatives. We study absolute and relative  {\em GROW\/} (`growth-rate optimal in the worst-case') e-variables as introduced by \cite{GrunwaldHK19},  and we show how such
e-variables are related to a concentration inequality which we call  {\em Csisz\'ar-Sanov-Chernoff\/} (CSC from now on). The 1-dimensional version of this inequality is well-known as a straightforward application of the Cram\'er-Chernoff method and  is sometimes called the {\em generic\/} Chernoff bound. The multivariate version was apparently first derived by \cite{Csiszar84} as a (significant) strengthening of Sanov's classical theorem; we review this history in Section~\ref{sec:cscconvex} beneath Theorem~\ref{thm:old}. Given all this we decided to name the bound  `Csisz\'ar-Sanov-Chernoff'. 

Formally, we consider a $d$-dimensional random vector $Y=(Y_1, \ldots, Y_d)$ supported on some (possibly finite or countable) subset $\cY$ of $\reals^d$. Whenever we speak of `a distribution for $Y$' we mean a distribution on $\cY$ equipped with its standard Borel $\sigma$-algebra. Then $\conv(\cY)$, the convex hull of $\cY$, is the set of means for such distributions; we invariably assume that the zero-vector  $(0,\ldots, 0)^{\top}$ (which we abbreviate to $\nv$ whenever convenient) is contained in $\cY$. We then let the null hypothesis $P_0$ be a distribution for $Y$  with mean equal to the zero vector: ${\mathbb E}_{P_0}[Y] = \nv$.
We fix a background measure $\rho$ and assume that $Y$ has a density $p_0$ relative to $\rho$ under $P_0$ (we may in fact take $\rho$ equal to $P_0$ itself without restricting the generality of our results).
We assume that we are given a set of means $\meanspace_1 \subset \conv(\cY)$ and an alternative (i.e. a set of distributions on $\cY$)  $\cH_1$ that is {\em compatible\/} with the given $\meanspace_1$, in the sense that, for all $\mu \in \conv(\cY)$:
\begin{equation}\label{eq:compatible}    
\mu \in \meanspace_1 \Leftrightarrow \text{\ there exists $P \in \cH_1$ with ${\mathbb E}_P[Y]= \mu$}.
\end{equation}
We consider various such $\cH_1$. In general, 
$\cH_1$ is allowed to contain distributions that do not have densities, but whenever a $P_1 \in \cH_1$ does have a density, it is denoted by small letters, i.e. $p_1$.
We further invariably assume that $P_0$ and $\cH_1$ are {\em separated\/} in terms of the mean. That is,
$
\inf_{\mu \in \meanspace_1 } \| \mu \|_2 > 0,
$
and finally, that $Y$ admits a moment generating function under $P_0$. This is a strong assumption, but it is the only strong assumption we impose. 

In Section~\ref{sec:convexm1} we consider the GROW (growth-rate-optimal in the worst-case over $\cH_1$)
e-variable $S_{\grow}$ for this scenario, assuming either that $\cH_1$ is the set $\cP_1$ of {\em all\/} $P$ with mean in some given {\em convex\/} set $\meanspace_1$, or that $\cH_1$ is the set $\cE_1$ of all elements of 
the exponential family generated by $P_0$ with means in $\meanspace_1$, or any $\cH_1$ with $\cE_1 \subset \cH_1 \subset \cP_11$ --- it turns out that the GROW e-variables coincide for all such $\cH_1$ . 
We show how this result can be derived using the celebrated Csisz\'ar-Tops{\o}e  Pythagorean theorem for relative entropy
and how it leads to the basic CSC concentration inequality. 
We do not claim novelty for this section, which mostly contains  re-formulations of results that are well-known in the information-theoretic (though perhaps not in the e-value) community. The real novelty comes in subsequent sections: 

In Section~\ref{sec:surrounding} we  move to the case that the {\em complement\/} of $\meanspace_1$ is a connected, bounded set containing $P_0$ --- a case that is more likely to arise in practical applications, is more closely related to the setting of the multivariate CLT, yet has, as far as we know, not been considered before when deriving CSC bounds, with the exception of \cite{kaufmann2021mixture} who consider a variation of this setting (we return to their results in the final Section~\ref{sec:discuss}). We denote this as the {\em surrounding\/} $\cH_1$ case, since $P_0$ is `surrounded' by $\cH_1$. 
We can extend the previous $S_{\grow}$ e-variable to this case in two ways. We may either look at the straightforward {\em absolute\/} extension of the GROW e-variable to the multivariate case, which we still denote by $\Sgrow$; or we can determine a {\em relatively\/} optimal GROW e-variable $\Srel$ that is as close as possible to the largest $\Sgrow$ among all e-variables $\Sgrow$ that can be defined on convex subsets of $\cH_1$, where, in this paper, as in \cite{jang2023tighter}, we define relative optimality in a minimax-regret sense. We characterize $\Sgrow$ for the case that $d=1$ (leaving the complicated case $d > 1$ as an open question), and we characterize $\Srel$ for general dimension $d$. 
We then show that $\Srel$ leads again to a CSC bound, Theorem~\ref{thm:nml} --- and this CSC bound is new. 

The CSC bound arrived at in Theorem~\ref{thm:nml} contains a minimax regret term $\mmreg$, which may be hard to verify in practice. In typical applications, we will have $Y = n^{-1}\sum_{i=1}^n X_i$ with $X_i$ i.i.d., for some fixed sample size $n$. Then, if the exponential family generated by $P_0$ is regular (as it will be in most cases), we know that $Y$ is equal to $\hat\mu_{|X_1, \ldots, X_n}$, the maximum likelihood estimator for the generated family, given in its mean-value parameterization. We can then think of the CSC bound as a concentration inequality that bounds the probability of the MLE falling in some set. Based on this instantiation of $Y$, we provide, in Section~\ref{sec:asymptotic_growth_rate},  based on earlier work by \cite{ClarkeB94,takeuchi1997asymptotically}, asymptotic expressions of the minimax regret term $\mmreg_n$ as a function of $n$, and show that, under regularity conditions on the boundary of the set $\meanspace_1$, it increases as
$$
\frac{d-1}{2} \log n + O(1), 
$$
It is no coincidence that this term is equal to the BIC/MDL model complexity for $d-1$-dimensional statistical family: it turns out that the boundary of $\meanspace_1$ is the relevant quantity here, and it defines a $(d-1)$-dimensional exponential family embedded within the $d$-dimensional family generated by $P_0$. 
We show how this result gives us an asymptotic expression for the absolute GROW e-variable $\Sgrow$ after all, provided that the complement of $\meanspace_1$ is a Kullback-Leibler ball around $P_0$. 

This paper is still a work in progress. In the final section we provide additional discussion of the results, a comparison to the multivariate Central Limit Theorem, and we indicate the future work we would like to add to our current results. 

\subsection{Background on GROW e-variables}
Since it will help provide the right context, in this --- and only in this --- subsection we allow composite null hypotheses $\cH_0$. Each $P \in \cH_0$ is then a distribution for $Y$. 

\begin{definition}
A nonnegative statistic $S= s(Y)$ is called an $e$-variable relative to $\mathcal{H}_0$ if 
$$
\text{for all $P \in \mathcal{H}_0$: } \mathbb{E}_P[S] \leq 1.
$$
\end{definition}
Let $\cS_0$ be the set of all e-variables that can be defined relative to $\cH_0$ and such that ${\mathbb E}_P[\log S]$ is well-defined as an extended real number for all $P \in \cH_1$, where we adopt the convention that $\log \infty = \infty$ and $\log 0 = \infty$. `Well-defined'  means that we may have ${\mathbb E}_P[(\log S) \vee 0] = \infty$  or ${\mathbb E}_P[(\log S) \wedge 0] = - \infty$ but not both.
\cite{GrunwaldHK19} defines the $worst$-$case\ optimal\ expected\ capital\ growth\ rate$ $(\grow)$ as
\begin{equation}
    \label{eq:growmax}
\grow
:=
\sup\limits_{S: S \in \mathcal{S}_0} \inf\limits_{P \in \cH_1} \mathbb{E}_{P}[\log S],
\end{equation}
where $\mathbb{E}_{P}[\log S]$ is the so-called $growth\ rate$ of $S$ under $P \in \cH_1$. The $\grow$ $E$-variable, denoted as $\Sgrow$, if it exists, is the e-variable achieving the supremum above. We refer to \cite{GrunwaldHK19,ramdas2023savi}
for extensive discussion on why this is, in a particular sense, the {\em optimal\/} e-variable that can be defined for the given testing problem. 
As a special case of their main result, \citet[Theorem 2] {GrunwaldHK19} show the following:
\begin{theorem}\label{thm:ghk}{\bf  \cite[Theorem 2, Special Case]{GrunwaldHK19}} Suppose that  (a) $D(P_1 \| P_0)< \infty$ for all $P_1 \in\cH_1, P_0 \in \cH_0$ and (b)
\begin{equation}
    \label{eq:growmin}
\min\limits_{P_1 \in \conv(\cH_1)} \min_{P_0 \in \conv(\cH_0)} D(P_1 \| P_0)
\end{equation}
is achieved by some $P_1^*, P_0$, then we have
\begin{align}\label{eq:GHK}
&    \sup_{S \in \cS_0} \inf_{P \in \cH_1} \mathbb{E}_{Y \sim P} \left[\log S \right]
    = D(P^*_1 \| P_0) = \grow = \inf_{P \in \cH_1} \mathbb{E}_{Y \sim P}\left[ 
    \log\frac{p^*_1(Y)}{p_0(Y)}\right],  \\ \label{eq:grow}
    & \text{\ and $\Sgrow$, achieving (\ref{eq:growmax}), is therefore given by\ } \Sgrow= \frac{p_1^*(Y)}{p_0(Y)}.
    \end{align}
Here $p_1^*$ is the density of $P_1^*$, which exists by the finite KL assumption.
\end{theorem}
\subsection{Simple $\cH_0$ and the Pythagorean Property}
\label{sec:pythagoras}
Most recent work in e-variable theory has concentrated on the case of composite 
$\cH_0$ and simple $\cH_1$ \citep{ramdas2023savi}. Throughout this paper we consider the reverse case, simple $\cH_0 = \{P_0\}$ and composite $\cH_1$.
Now the problem clearly simplifies and in fact, a lot more has been known about this special case since the 1970s, albeit expressed in the different language of data-compression: in a landmark paper, \citet[Theorem 9]{Topsoe79} proved a minimax result for relative entropy  which  (essentially) implies (\ref{eq:GHK}) for the simple $\cH_0$ case. In fact, his result even implies that a distribution $P^*_1$ such that (\ref{eq:GHK}) and  (\ref{eq:grow}) hold exists under much weaker conditions, in particular condition (b) above is not needed: $P_1^*$ exists even if the  minimum in $\min_{P \in \conv(\cH_1)} D(P_1 \| P_0)$ is not achieved. 
The key result that Tops{\o}e used to prove his version of (\ref{eq:GHK}) and (\ref{eq:grow}) is Theorem 8 of his paper, a version of the {\em Pythagorean theorem\/}
for KL divergence originally due to Csisz\'ar \citep{Csiszar75,CoverT91,csiszar_information_2003} . We will re-state this result and explicitly use it to re-derive a version of (\ref{eq:GHK}) and (\ref{eq:grow}) that is slightly stronger than Tops{\o}e's and better suited to our needs (Tops{\o}e's Theorem 9 still assumes condition (a); our derivation weakens it). 

The  Pythagorean theorem expresses that in the following sense, the KL divergence behaves like a squared Euclidean distance: for arbitrary $P_0$ and  $\cH_1$ as above, we have as long as $\cH_1$ is {\em convex} and $\inf_{P_1\in \cH_1} D(P_1 \| P_0) < \infty$, that
there exists a probability distribution $P^*_1$, called the {\em information projection of $P_0$ on $\cH_1$}, that satisfies: 
\begin{align}
    \label{eq:pythagoras}
& \text{for all $P \in \cH_1$:}\  
D(P \| P_0) \geq D(P \| P^*_1) + D(P^*_1 \| P_{0}) \\ \nonumber
&\text{for every $Q_1, Q_2, \ldots \in \cH_1$ 
with $\lim_{j \rightarrow \infty} D(Q_j \| P_0) =
\inf_{P \in \cH_1} D(P_1 \| P_0)
$, we have}: \\& 
\ \ \ \ \lim_{j \rightarrow \infty} D(Q_j \| P^*_1) = 0.
\\ \label{eq:railjet}
& D(P^*_1 \| P_0) \leq \inf_{P \in \cH_1} D(P \| P_0).
\end{align}
In standard cases, the final inequality will hold with equality; in particular we have equality if  $\min_{P_1 \in \cH_1} D(P_1 \| P_0)$ is achieved. 

We call (\ref{eq:pythagoras}) the {\em Pythagorean property}. Note that it is {\em implied\/} by convexity of $\cH_1$ and finiteness of $\inf D(P_1 \| P_0)$, but it may sometimes hold even if $\cH_1$ is not convex. 

We now show, slightly generalizing Tops{\o}e's result, how the Pythagorean property (\ref{eq:pythagoras}) implies a version of 
\cite{GrunwaldHK19}'s theorem for simple $\cH_0$ (in fact, in the reformulation as a minimax theorem for data compression, the Pythagorean property is in fact {\em equivalent\/} to the minimax statement but we will not need that fact here; see \cite[Section 8]{GrunwaldD04} for an extended treatment of this equivalence).

\begin{proposition}\label{prop:pythagorasminimax}{\bf [Pythagoras $\Rightarrow S_{\grow} = p_1^*/p_0$]}
Suppose that $\cH_0 = \{P_0\}$ and let $\cH_1$ be any set of distributions (not necessarily convex!) for $Y$ such that $\inf_{P \in \cH_1} D(P \|  P_0) < \infty$, and suppose that $P_1^*$ is such that (\ref{eq:pythagoras})--(\ref{eq:railjet})  holds, so that it has a density $p^*_1$.  Further assume that, with `well-defined' defined as above (\ref{eq:growmax}), 
\begin{equation}\label{eq:THEcondition}
{\mathbb E}_P[\log p^*_1(Y)/p_0(Y)] \text{\rm \   is well-defined for all $P \in \cH_1$},
\end{equation}
so that $p^*_1(Y)/p_0(Y) \in \cS_0$. Then we have:
\begin{equation}\label{eq:pythagorasminimax}
    \sup_{S \in \cS_0} \inf_{P \in \cH_1} \mathbb{E}_{Y \sim P} \left[\log S \right]
    = D(P^*_1 \| P_0) = \inf_{P \in \cH_1} \mathbb{E}_{Y \sim P}\left[ 
    \log\frac{p^*_1(Y)}{p_0(Y)}\right]
    \end{equation}
so that $S_{\textsc{grow}} = p^*_1(Y)/p_0(Y)$.     
\end{proposition}
\begin{proof}
Since we deal with a simple null, it holds that (as shown by \cite{GrunwaldHK19}) any $S \in \cS_0$ must be of the form $q(Y)/p_0(Y)$ for some sub-probability density $q$ relative to the measure $\rho$, and any such ratio defines an e-variable: the notions are equivalent. Here `sub-probability' means that $\int q(y) d \rho(y) \leq 1$ is allowed to be smaller than $1$.
We thus have, with $\sup_q$ denoting the supremum of all sub-probability density functions $q$ for $Y$, 
    \begin{align}
        \sup_{S \in \cS_0} \inf_{P \in \cH_1} \mathbb{E}_{Y \sim P} \left[\log S \right]
        & = \sup_{q} \inf_{P \in \cH_1} \mathbb{E}_{Y \sim P} \left[\log \frac{q(Y)} {p_0(Y)} \right] \nonumber \\ \label{eq:leftbound}
        & \leq \sup_q \mathbb{E}_{Y \sim P^*_1} \left[\log \frac{q(Y)} {p_0(Y)} \right]  = D(P^*_1 \| P_0).
    \end{align}
    At the same time, the Pythagorean inequality (\ref{eq:pythagoras}) gives, by simple re-arranging of the logarithmic terms, that for all $P \in \cH_1$ with $D(P \| P_0) < \infty$: 
    \begin{align}\label{eq:obb}
            \mathbb{E}_{Y \sim P}\left[ 
    \log \frac{p^*_1(Y)}{p_0(Y)} 
    \right] \geq \mathbb{E}_{Y \sim P^*_1}\left[ 
    \log \frac{p^*_1(Y)}{p_0(Y)} 
    \right] = D(P^*_1 \| P_0),
    \end{align}
    whereas if $D(P \| P_0) = \infty$ then by assumption (\ref{eq:railjet}), the right-hand side of (\ref{eq:obb}) is finite. (\ref{eq:pythagoras}) then implies $D(P \| P_1) = \infty$, and. because by assumption (\ref{eq:THEcondition}), the left-hand side of (\ref{eq:obb}) is well-defined, we can again re-arrange (\ref{eq:pythagoras}) to give (\ref{eq:obb}). Thus, we have shown that for all $P \in \cH_1$, (\ref{eq:obb}) holds. But then
     \begin{align}\label{eq:rightbound}
        \sup_{S \in \cS_0} \inf_{P \in \cH_1} \mathbb{E}_{Y \sim P} \left[\log S \right]
        & \geq  \inf_{P \in \cH_1} \mathbb{E}_{Y \sim P} \left[\log \frac{p^*_1(Y)} {p_0(Y)} \right] \geq D(P^*_1 \| P_0). 
    \end{align}
    Together, (\ref{eq:leftbound}) and (\ref{eq:rightbound}) imply the result.
  \end{proof}
While Condition (\ref{eq:THEcondition}) may look complicated, it is immediately verified to hold if $D(P_1 \| P_0) < \infty$ for all $P_1 \in \cH_1$ but also under Condition ALT-$\cH_1$ presented in the next section, which allows for $\cH_1$ to even contain distributions $P_1$ with $P_1 \not \ll P_0$ (see Example~\ref{ex:welldefined} for an instance of this). 

\subsection{The R\^ole of Exponential Families}
We shall  from now on tacitly assume that the convex support of $P_0$ is $d$-dimensional (see \cite[Chapter 1]{Brown86} for the precise definition of `convex support'). This is without loss of generality: if $Y$ takes values in $\reals^d$ yet the convex support does not have dimension $d$,  it must have dimension $d' < d$, and then we can replace $Y$ by $d'$-dimensional $Y'$ that is an affine function of $Y$ and work with $Y'$ instead.  
Combined with our earlier assumption that $Y$ has a moment generating function under $P_0$, it follows \citep[Chapter 1]{Brown86} that 
$Y$ and $P_0$ jointly {\em generate\/} a $d$-dimensional natural exponential family $\cE = \{P_{\vec{\theta}}: \vec{\theta} \in \Theta \}$: a set of distributions for $Y$  with parameter space $\Theta \subseteq {\mathbb R}^d$ . Each distribution $P_{\theta}$ has density $p_{\theta}$ relative to $\rho$, given by:
\begin{equation}\label{eq:expfam}
p_{\theta}(Y) = \frac{1}{Z(\theta)} \exp\left({\theta}^\top Y  \right) \cdot p_0(Y),
\end{equation}
where 
$Z(\theta)$ is the normalizing factor and $p_0$ is the density of the {\em generating distribution} $P_0$ and $\Theta = 
\{\theta: Z(\theta) < \infty\}$. 
From now on, we freely use standard properties, terminology and definitions concerning exponential families (such as  `carrier density' and so on), that can be found, in, for example, \cite{BarndorffNielsen78,efron_2022,Brown86}. We will only mention these works, and then specific sections therein, when we refer to results that are otherwise hard to find.

Parameterization (\ref{eq:expfam}) is called the {\em canonical\/} or {\em natural\/} parameterization. As is well-known, exponential families can be re-parameterized in terms of the mean of $Y$. Thus, there is a 1-to-1 mapping $\mu: \Theta \rightarrow \meanspace$, mapping each $\theta \in \Theta$ to $\mu(\theta) := {\mathbb E}_{P_{\theta}}[Y]$, with $\meanspace$ being the {\em mean-value parameter space}. We let $\theta(\mu)$ be the inverse of this mapping and let
$\Pmv_{\mu} := P_{\theta(\mu)}$ with density $\pmv_{\mu}(Y) := p_{\theta(\mu)}(Y)$. Then we can equivalently write our exponential family as 
\begin{equation}\label{eq:expfammean}
\cE = \{\Pmv_\mu: \mu \in \meanspace \}. 
\end{equation}
Without loss of generality, we assume that $Y$ is defined such that $\nv \in \meanspace$ and the natural 
parameterization is such that $\theta(\nv) = \nv, \mu(\nv) = 0$. 
Clearly the null hypothesis $\cH_0 = \{P_0\}$ is given by the element of $\cE$ corresponding to the parameter vector $\theta = \nv$.
The mean-value parameterization will be the most `natural' one (no pun intended) to use to define the alternative. 
As said in the introduction, we restrict ourselves to cases in which $Y$ has a moment generating function under $P_0$. Then $\meanspace$ contains an open set around $0$ and the exponential family (\ref{eq:expfammean}) exists.
We define the alternative $\cH_1$ in terms of a given set of means $\meanspace_1 \subset \conv(\cY)$, invariably satisfying:

\paragraph{Condition ALT-$\meanspace_1$:}
(a) $\meanspace_1$ is closed, and $\inf_{\mu \in \meanspace_1} \| \mu \|_2 > 0$; and, (b), for all $\mu \in \meanspace_1$, there is a $\mu' \in \meanspace \cap \meanspace_1$ that lies on the straight line connecting $\nv$ and $\mu$. \\ \ \\ \noindent
For the actual alternative hypothesis we then invariably further assume:
\paragraph{Condition ALT-$\cH_1$:}
(a) $\cH_1$ and $\meanspace_1$ are {\em compatible\/} in the sense of (\ref{eq:compatible}), and (b) $\cE_1 := \{ \Pmv_{\mu}: \mu \in \meanspace_1 \cap \meanspace\} \subset \cH_1$ and, (c), for all $P \in \cH_1$, ${\bf E}_{Y \sim P}[Y]$ is well-defined. \\ \ \\ \noindent
To appreciate these conditions, consider first the case that $\meanspace= \conv(\cY)$, i.e. the mean-value parameter space of family $\cE$ contains every possible mean. Then Condition ALT-$\meanspace_1$ (a)  says that $\meanspace_1$ is separated from $0$ and that it contains its boundary (note that it does not need to be bounded: for example, in the case that $\cE$ is the 1-dimensional normal location family, having $\meanspace_1= [1,\infty)$ is perfectly fine). Condition ALT-$\meanspace_1$ (b) holds automatically if $\meanspace= \conv(\cY)$ (e.g. in the Gaussian location case); the example below illustrates the case that $\meanspace$ is a strict subset of $\conv(\cY)$. Condition ALT-$\cH_1$ (b) simply says that for every mean in $\meanspace_1$, $\cH_1$ contains the element of the exponential family $\cE$ with that mean. 
\begin{example}\label{ex:ample} {\rm 
As a very simple example, suppose that $Z_1, Z_2, \ldots$ are i.i.d. Bernoulli$(p)$ for some $0 < p < 1$, and $X_i = 1/p$ if $Z_i=1$ whereas $X_i = - 1/(1-p)$ if $Z_i=0$. Let  $Y=n^{-1} \sum_{i=1}^n X_i$. Then  $\conv(\cY) = [-1/(1-p),1/p]$ and according to $P_0$, $Y$ is a linear transform of a $\textsc{bin}(n,p)$ random variable with ${\mathbb E}_{P_0}[Y]= 0$. Then Condition ALT-$\meanspace_1$ (a) expresses that $\meanspace_1$ must not contain a neighborhood of  the mean of $P_0$ (i.e., $\mu=0$), and (b) that it must not be restricted to singletons at the boundary (i.e. $\mu = 1/p$ or $\mu= -1/(1-p)$). However, for example, $\meanspace_1=[1/p-\epsilon,1/p]$ for any $0 < \epsilon < 1/p$ satisfies Condition ALT-$\meanspace_1$. 
We see that the condition only very minimally restricts the set of $\meanspace_1$ that are allowed. It becomes more restrictive for the (very seldomly encountered!) case that the generated exponential family $\cE$ is {\em irregular\/} \citep{BarndorffNielsen78}. By construction of $\cE$, this is equivalent to it being {\em not steep}. For example, let $Y$ be 1-dimensional and let  $P_0$ have density $p_0(y) = {\bf 1}_{y > 1} \cdot (2/y^3)$ relative to Lebesgue measure. Then we get  $\cE = \{P_{\theta} \mid \theta \leq 0\}$ with $p_{\theta}(y) \propto \exp(\theta  y)p_0(y)$. Then $\meanspace= (1,2]$ yet $\cY= (1,\infty)$ (from the fact that $\meanspace$ is not open we immediately see that $\cE$ is not regular). Condition ALT-$\meanspace_1$ now requires that $\meanspace_1$ contains a $\mu < 2$, even though $P_0(Y \geq b)> 0$ for any $b \geq 2$. }
\end{example}
\commentout{
\paragraph{Condition 1}
We assume that:
\begin{enumerate}
    \item $Y$ has a moment generating function under $P_0$, so $P_0$ and $Y$ generate an exponential family $\cE = \{\Pmv_{\mu}: \mu \in \meanspace\}$. 
    \item We assume that this family is {\em regular}, so (by \cite[Theorem 9.2]{BarndorffNielsen78}), 
    $\meanspace$ is convex and equal to the interior of the convex hull of the support of $Y$ under $P_0$. 
    \item We assume that $\meanspace_1$ is contained in the interior of the convex hull of the support of $Y$ under $P_0$.
\end{enumerate}
The strong condition here is that a moment generating function exists, which implies that $Y$ has exponentially small tails. Once this is the case, in `most' cases the family generated by $P_0$ and $Y$ is steep, and therefore regular; see \cite{BarndorffNielsen78} who provides lots of examples and discussion. 
}

\section{Convex $\meanspace_1$}
\label{sec:convexm1}
\paragraph{The Connection between Exponential Families and the Pythagorean Theorem}
Although exponential families are usually employed as families that are reasonable in their own right, as is well-known \citep{Csiszar75,GrunwaldD04} they can also be arrived at as characterizing  the information projection $P^*_1$ in the Pythagorean property above for certain $\cH_1$. We will heavily use this characterization below. 
Variations of the following result (see Figure~\ref{fig:enter-label} for illustration) are well-known:
\begin{proposition}\label{prop:exp}{\bf [GROW e-variable is $\bar{p}_{\mu^*}/p_0$, with $\bar{P}_{\mu^*}$ an element of exponential family $\cE$
]}
    Let $\cH_0 =\{P_0\}$ and $\meanspace_1$ be such that Condition ALT-$\meanspace_1$  holds. Furthermore let $\meanspace_1$ be convex. Then there exists $\mu^* \in \meanspace \setminus \{0 \}$  uniquely achieving $\inf_{\mu \in \meanspace_1} D(\Pmv_{\mu} \| P_0)$, and we have: 
\begin{equation}\label{eq:rightmin}
\min_{\mu \in \meanspace_1 \cap \meanspace} D(\Pmv_{\mu} \| P_0) = 
\min_{\mu \in \meanspace_1} D(\Pmv_{\mu} \| P_0) = 
D(\Pmv_{\mu^*} \| P_0)  
        =         \theta^{*\top} \mu^* - \log Z(\theta^*) 
\end{equation}
with $\theta^* := \theta(\mu^*) \in \Theta$.
    Furthermore let $\cH_1$ be convex and such that Condition ALT-$\cH_1$ holds. Then the minimum in (\ref{eq:rightmin}) further satisfies: 
    \begin{align}\label{eq:eu}
         \inf_{P \in \cH_1} D(P \| P_0) 
        = D(\Pmv_{\mu^*} \| P_0)  
    \end{align}
    the minimum KL on the left being achieved uniquely by $\Pmv_{\mu^*}$. 
    As a consequence, 
$\Sgrow = \pmv_{\mu^*}(Y)/\pmv_0(Y) \in \cS_0$
and $\grow = D(\Pmv_{\mu^*} \| P_0)$.
 \end{proposition}
 \begin{proof}
The KL divergence $D(\bar{P}_{\mu} \| P_0)$ is continuous in $\mu$, has its overall minimum over $\conv(\cY)$ in the point $\mu =0$ and is strictly convex. This implies (i) that by Condition ALT-$\meanspace_1$(a), $\min_{\mu \in \meanspace_1} D(\bar{P}_{\alpha\mu} \| P_0)$ is uniquely achieved for some $\mu^*$ on the boundary of $\meanspace_1$. By Condition ALT-$\meanspace_1$(b), the boundary of $\meanspace_1$ is included in $\meanspace$, so $\mu^* \in \meanspace_1 \cap \meanspace$. 
This yields the first two equations in (\ref{eq:rightmin}). Writing out the densities in $D(\Pmv_{\mu^*} \| \Pmv_0)$ then gives the rightmost equality.

It remains to prove (\ref{eq:eu}).
Condition ALT-$\cH_1$ implies that $\cH_1$ contains a $\bar{P}_{\mu} \in \cE$, hence $\meanspace_1$ contains a $\mu \in \meanspace$,  and since $D(\bar{P}_{\mu} \|P_0) < \infty$ for all $\mu \in \meanspace$, we have  $\inf_{P \in \cH_1} D(P \| P_0) < \infty$. Therefore, it suffices to show (\ref{eq:eu}) with the infimum taken over $\{P \in \cH_1: D(P \| P_0) < \infty \}$. In particular all $P$ in this set have a density $p$.  
Thus, fix any $P$ in the set $\{P \in \cH_1: D(P \| P_0) < \infty \}$ and let $\mu = \mathbb{E}_{P}[Y]$.
We first consider the case that $\mu \in \meanspace$, so that $\Pmv_{\mu} \in \cE$ (in particular then also $\mu \in \meanspace \cap \meanspace_1$ and $\Pmv_{\mu} = P_{\theta}$ with $\theta= \theta(\mu)$; note though that we may have $P \neq \bar{P}_{\mu}$). 
Straightforward rewriting and linearity of expectation gives 
\begin{align}\label{eq:sport}
&     D(P \| P_0)= \mathbb{E}_{Y \sim P}\left[
\log \frac{p(Y)}{p_0(Y)}
    \right] \geq \mathbb{E}_{P}\left[
\log \frac{p_{\theta}(Y)}{p_0(Y)}
    \right] = 
     \mathbb{E}_{P}\left[
\log \frac{1}{Z(\theta)} \cdot e^{\theta^{\top} Y } 
    \right]=  \nonumber \\ & 
   \theta^{\top} \mu - \log Z(\theta) = D(P_{\theta} \| P_0)= D(\Pmv_{\mu} \| P_0)  
 \geq \min_{\mu \in \meanspace_1 \cap \meanspace} D(\Pmv_{\mu} \| P_0),
\end{align}
the final inequality following because $\mu \in \meanspace \cap \meanspace_1$. Together with (\ref{eq:rightmin}) this shows (\ref{eq:eu}) for the case that $\mu \in \meanspace$.
It remains to consider the case that  $\mu \not \in \meanspace$. In that case, Condition ALT-$\meanspace_1$ and ALT-$\cH_1$ imply that there exists a $\mu' \in \meanspace \cap \meanspace_1$ and $\Pmv_{\mu'} \in \cH_1 \cap \cE$ such that $\mu' = \alpha \mu$ for some $0 < \alpha < 1$. Retracing the steps of (\ref{eq:sport}) with $\theta' = \theta(\mu')$ in the place of $\mu$, we find 
\begin{equation}\label{eq:doei}
 D(P \| P_0) \geq 
   \theta^{'\top} \mu - \log Z(\theta') = f(1)
\end{equation}
where, for $\gamma \in [0,1]$, we set $f(\gamma) = {\mathbb E}_{P_{\gamma \mu}}\left[\log \frac{\pmv_{\mu'}(Y)}{p_0(Y)} \right]$.
Since $f(0)$ is minus a KL divergence, $f(0) < 0$. Also, $f(\alpha \mu) = f(\mu') > 0$, since $f(\mu')$ is a KL divergence. Since $f(\gamma)$ is linear in $\gamma$, it follows that $f(\gamma)$ is strictly increasing so $f(1) > f(\mu')$ and then (\ref{eq:doei}) gives that $D(P \|P_0) \geq D(\Pmv_{\mu} \| P_0)$ which again implies the result. 

It remains to prove that $S_{\grow}  := S$ with $S= \bar{p}_{\mu^*}(Y)/p_0(Y)$. For this, first  note that 
for all $P \in \cH_1$, we have $P(p_0(Y) > 0)=1$ by definition (namely, $\cY$ is the support of $Y$ under $P_0$), from which it follows that $P(S > 0) =1$,
so that $\log S= \theta^{*\top} Y - \log Z(\theta^*)$ whence ${\bf E}_{Y \sim P}[\log S]$ is well-defined; it follows that 
$S\in \cS_0$. By Tops{\o}e's result and the assumed convexity of $\cH_1$ and finiteness of $D(\bar{P}_{\mu^*} \| P_0)$, we may now apply  Proposition~\ref{prop:pythagorasminimax}, and the result follows. 
\end{proof}
 \begin{example}{\rm 
    \label{ex:welldefined}
Since $\cE$ is an exponential family, we know that all elements $P \in \cE$ have the same support as $P_0 \in \cE$, and, by definition of $P_0$, this support is equal to $\cY$. This implies that, even if  $P \in \cH_1$ puts positive mass on an outcome $y \in \cY$ that has mass $0$ under $P_0$,   then (because $y$ must be in $P_0$'s support), well-definedness (\ref{eq:THEcondition}) may still hold. 
For example, consider the case that  $Y = \reals$ and $P(\{0 \}) = 1/2$ and $P \mid Y \neq 0 = N(0,1)$ is a standard normal, and $\cE$ is the normal location family so that $\bar{P}_{\mu} = N(\mu,1)$. We get ${\mathbb E}_P[\log \bar{p}_{\mu}(Y)/p_0(Y)]= (1/2) D(P_0 \| \bar{P}_{\mu})+ \mu^2/2$, i.e. it is well-defined. On the other hand, if we were to allow $P_0$ defined on a sample space $\cY$ with $Y$'s support under $P_0$ a strict subset of $\cY$, and we would take  $P \in \cH_1$ that put positive mass on an outcome that is not in the support of $P_0$, then $\bar{p}_{\mu}(Y)/p_0(Y)$ would evaluate to $0/0$ with positive $P$-probability, and $\Sgrow$ of the form above would be undefined. We avoid such issues by requiring $\cY$ to coincide with its support under $P_0$.  We suspect that  using the ideas of \cite{larsson2024numeraire}, we can even obtain well-defined growth expressions for this case, but will leave this for future work. }
\end{example}

\subsection{CSC (Chernoff-Sanov-Csisz\'ar) for convex $\meanspace_1$}
\label{sec:cscconvex}
Note that the only role $\cH_1$ plays in the theorem below is to make $\Sgrow$ well-defined; the bounds further do not depend on the specific choice of $\cH_1$ as long as $\Sgrow = \pmv_{\mu^*}(Y)/p_0(Y)$.

\begin{figure}
    \centering
    \includegraphics[width=0.5\linewidth]{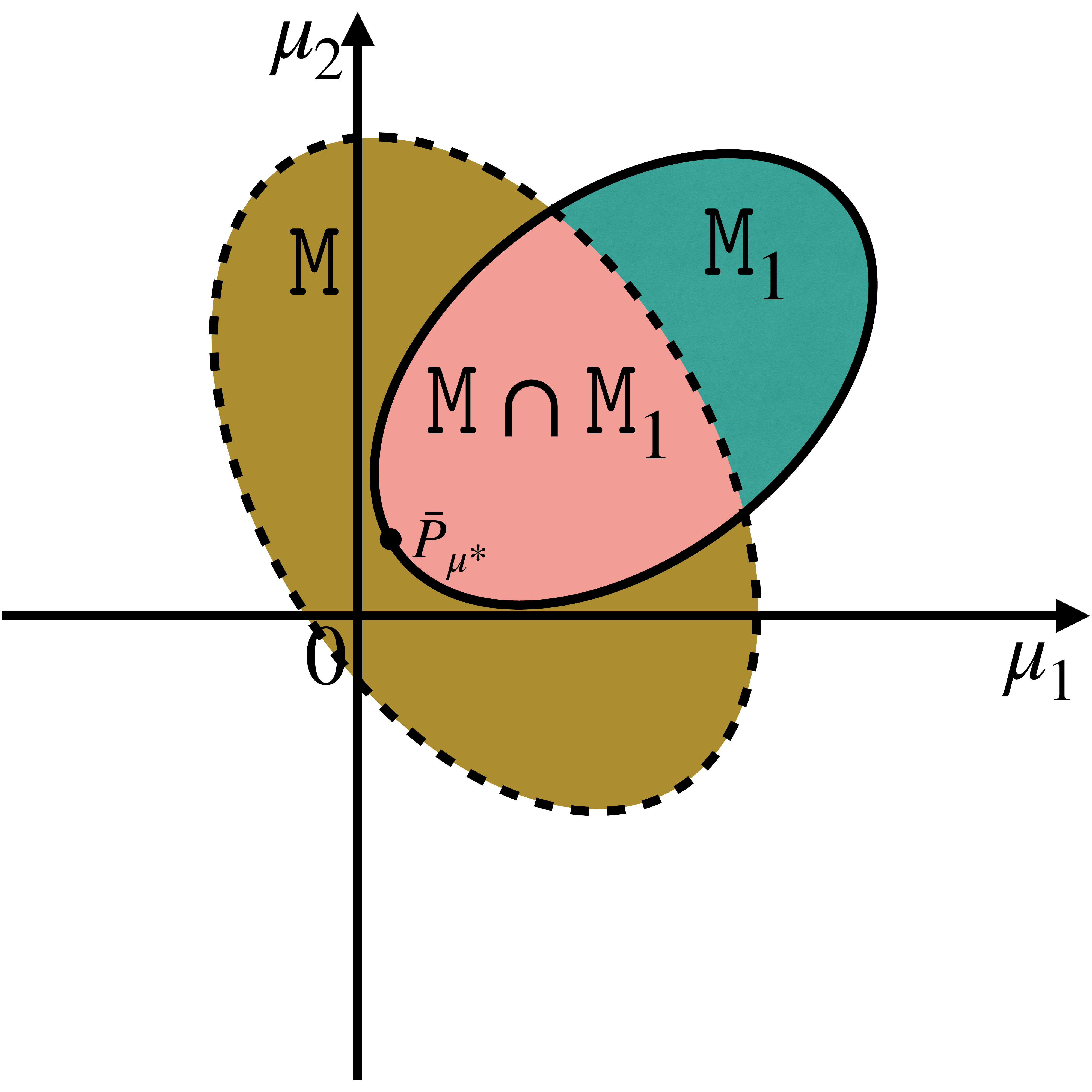}
    \caption{Convex $\mathtt{M}_1$. $\mathtt{M}_1$ is the mean parameter set that  $\cH_1$ is compatible with, and $\mathtt{M}$ is the mean parameter space of the exponential family generated from $P_0$. The setting of this  figure satisfies Condition ALT-$\meanspace_1$.}
    \label{fig:enter-label}
\end{figure}

\begin{theorem}\label{thm:old} Suppose $P_0$ and $\meanspace_1$ are such that $\meanspace_1$ is convex and Condition ALT-$\meanspace_1$ holds so that, by Proposition~\ref{prop:exp}, there exists $\mu^* \in \meanspace$ minimizing $D(\bar{P}_{\mu} \| P_0)$ over $\meanspace_1$ with $\Pmv_{\mu^*} \in \cE$. Let $\cH_1$ be any set of distributions such that Condition ALT-$\cH_1$ holds, so that, by Proposition~\ref{prop:exp},  $\Sgrow(Y) := \pmv_{\mu^*}(Y)/p_0(Y)$. Define 
\begin{align}
    \underline{D}:= \inf_{\mu \in \meanspace_{1}}  D(\Pmv_{\mu} \| P_0). 
\end{align}
We have:
\begin{equation}\label{eq:subset}
y \in \meanspace_{1} \Rightarrow \Sgrow(y) \geq  e^{\underline{D}}
\end{equation}
so that
$$
{\mathbb E}_{Y \sim \Pmu_0}\left[ {\bf 1}_{Y \in \meanspace_{1}}
\cdot {e^{\underline{D}} } \right]
\leq {\mathbb E}_{Y \sim \Pmu_0}\left[ {\bf 1}_{Y \in \meanspace_{1}}
\cdot {S_{\grow}} \right] \leq 
{\mathbb E}_{Y \sim \Pmu_0}\left[ {S_{\grow}} \right] 
\leq 1.
$$
As a consequence,  we have:
\begin{align}\label{eq:basiccsc}
 P_0(Y \in  \meanspace_{1}) 
 = {\mathbb E}_{Y \sim \Pmu_0}\left[ {\bf 1}_{Y \in \meanspace_{1}}
\cdot {e^{\underline{D}} } \right] \cdot e^{- \underline{D}}
 \leq   e^{-  \underline{D}},
\end{align}
and we also have, for the one-dimensional case with $Y \in \reals$, for any $D>0$ for which there exists $\mu^* \in \meanspace$ such that $D(\Pmv_{\mu^*} \| P_0) =  D$,  with $s = \sgn(\mu^*)$,
\begin{align}
\label{eq:MLEbound}
& P_0\left( \frac{\sup_{\mu \in \meanspace} \bar{p}_{\mu}(Y)}{p_0(Y)} \geq e^{D}, \textsc{sgn}(Y) = s \right) 
= P_0\left( \frac{\bar{p}_{\mu^*}(Y)}{p_0(Y)} \geq e^{D} \right)
\leq   e^{- D}.
\end{align}
\end{theorem}

(\ref{eq:basiccsc}) is the bound first developed by \cite{Csiszar84}, who presented it as an extension of part of Sanov's theorem. In the one-dimensional case, it can also be seen as the `generic' Chernoff bound. 
This bound is usually formulated as 
$$
``P_0(Y\geq \mu^*) \leq \inf_{\theta > 0} {\mathbb E}_{P_0}[e^{\theta^{\top} Y}]e^{- \theta^{\top} \mu^*}."$$ 
Evaluating the infimum shows that the right-hand side is equal to $\exp(-\underline{D})$ with $\underline{D}= D(\bar{P}_{\mu^*} \| P_0)= \inf_{\mu \geq \mu^*} D(\bar{P}_{\mu} \| P_0)$, making the bounds equivalent; hence our choice of the {\em CSC}-terminology (this is the {\em generic\/} Chernoff bound; when authors speak of {\em the\/} Chernoff bound they usually refer to a specific instance of it, with $Y= \sum_{i=1}^n X_i$ and $X_i$ binary). 
Links between Sanov's theorem, Csisz\'ar's extension thereof and Chernoff have been explored before; see e.g. Van Erven [\citeyear{VanErven12}].

We also note that, while versions of  (\ref{eq:MLEbound}) have been known for a long time, it is sometimes considered surprising, because a direct, more naive application of Markov's inequality would give, with $
S' = \frac{\sup_{\mu \in \meanspace}\bar{p}_{\mu}(Y)}{p_0(Y)}$,
$$
P_0( S' \geq e^{D}) \leq e^{-D} \cdot  {\mathbb E}_{P_0}\left[S' \right]
= e^{-D} \cdot \int p_0(y) \cdot \frac{\sup_{\mu \in \meanspace}\bar{p}_{\mu}(y)}{p_0(y)} d \rho(y) \gg e^{-D},
$$
which can be considerably weaker, since $S'$ is not an e-variable. The result (\ref{eq:MLEbound}) shows that we {\em can\/} establish an underlying e-variable, and it is given by $\bar{p}_{\mu^*}/p_0$. 

Even though Theorem~\ref{thm:old} is not new, we give its proof in full since its ingredients will be reused later on:
\begin{proof}{\bf [of Theorem~\ref{thm:old}]}
Let $\meanspace_1$ and $\cH_1$ be as in the theorem, so that  ALT conditions hold.  Note that we may choose $\cH_1$ to be convex. By the final part of Proposition~\ref{prop:exp} we then  know that  $\Sgrow = \frac{\pmv_{\mu^*}({Y})}
{p_{0}({Y})}$ and that for all $P \in \cH_1$, with $\mu ={\mathbb E}_P[Y]$ and $\underline{D} = D(\Pmv_{\mu^*} \| P_0)$, we have 
\begin{equation}\label{eq:meanimp}
\mu \in \meanspace_1 \Rightarrow {\mathbb E}_{P} \left[ \log \frac{\pmv_{\mu^*}({Y})}
{p_{0}({Y})} \right] \geq  \underline{D},
\end{equation}
where the expectation is well-defined. 
Now note that the right-hand side can be rewritten, with $\theta^* = \theta(\mu^*)$, as 
$$
\theta^{*\top} \mu - \log Z(\theta^*) \geq  \underline{D}.
$$
Thus
(\ref{eq:meanimp}) can be rewritten as:
$$
y \in \meanspace_1 \Rightarrow  \theta^{*\top} y - \log Z(\theta^*)  \geq  \underline{D},
$$
or, again equivalently,
$$
y \in \meanspace_1 \Rightarrow \log \pmv_{\mu^*}(Y)/\pmv_0(Y)\geq  \underline{D}.
$$
The result (\ref{eq:subset}) follows after exponentiating.
The subsequent inequality (\ref{eq:basiccsc}) now readily follows as well.

For (\ref{eq:MLEbound}), we only consider the case $\mu^* > 0$, the case $\mu^* < 0$ being completely analogous. We have
\begin{align}\label{eq:mannheim}
   &  P_0\left(\log \sup_{\mu \in \meanspace} 
    \frac{\bar{p}_{\mu}(Y)}{p_0(Y)} \geq {D}, Y \geq \mu^* \right)=  P_0\left(Y \geq \mu^* \right) =
    \nonumber \\
   &  P_0\left(
    \frac{\bar{p}_{\mu^*}(Y)}{p_0(Y)} \geq e^{D}, Y \geq \mu^* \right)\leq P_0(S_{\grow} \geq D) \leq e^{-D},
\end{align}
which follows because, by the {\em robustness property of exponential families\/} \citep[Chapter 19]{grunwald2007minimum}(but also easily verified directly by considering the natural parameterization), $\log p_{\mu^*}(y)/p_0(y) = D(\bar{P}_{\mu^*} \| P_0)$ if $y= \mu^*$, and $\log p_{\mu^*}(y)/p_0(y)$ is increasing in $y$ for $\mu^* > 0$, which implies that the events inside the left and the right probability are identical. 
On the other hand,

\begin{align}\label{eq:kaiserslautern}
  &  P_0\left(\log \sup_{\mu \in \meanspace} 
    \frac{\bar{p}_{\mu}(Y)}{p_0(Y)} \geq {D}, 0 \leq Y < \mu^*) \right)= P_0\left(\log 
    \frac{\bar{p}_Y(Y)}{p_0(Y)} \geq {D}, 0 \leq Y < \mu^*) \right)= \nonumber \\ 
    & P_0 (D(\bar{P}_Y \|P_0) \geq D(\bar{P}_{\mu^*} \|P_0), 0 \leq Y < \mu^*) = 0, 
\end{align}
where the first equality follows because the $\mu$ maximizing the likelihood $\bar{p}_{\mu}(Y)$ is uniquely given by $Y$ if $Y \in \meanspace$, and the second is again the robustness property of exponential families, and the third follows because KL divergence $D(\bar{P}_{\mu} \| P_0)$ is strictly increasing in $\mu$ if $\mu > 0$. 

Together, (\ref{eq:mannheim}) and (\ref{eq:kaiserslautern}) imply the result. 
\end{proof}

\section{Surrounding $\meanspace_1$}
\label{sec:surrounding}
We now consider tests and concentration bounds  in the often more relevant setting of `surrounding' $\meanspace_1$.
Formally, we call $\meanspace_1$ {\em surrounding\/} if its complement, $\meanspace_1^{\comp} := \conv(\cY) \setminus \meanspace_1$ is an open, bounded, connected set containing $\nv$, and contained in $\reals^d$. We will call surrounding $\meanspace_1$ {\em nice\/} if (a) $\meanspace_1^{\comp}$ is contained in the interior 
of the mean-value space $\meanspace$ of exponential family $\cE$ generated by $P_0$ and also (b) $\meanspace_1^{\comp}$ is {\em $\nv$-star-shaped}, which means that for any straight line going through $\nv$, its intersection with $\meanspace_1^{\comp}$ is an interval, so that it crosses the boundary $\bd(\meanspace_1^{\comp})$ only once. Note in particular that any convex $\meanspace_1^{\comp}$ is automatically star-shaped; see Figure~\ref{fig:starA} and~\ref{fig:starB} for two examples of star-shaped $\meanspace_1^{\comp}$. Note also that any {\em nice\/} $\meanspace_1$ automatically satisfies Condition ALT-$\meanspace_1$. 
\begin{figure}
	\begin{minipage}[t]{0.5\linewidth}
		\centering
        \includegraphics[width=2.2in]{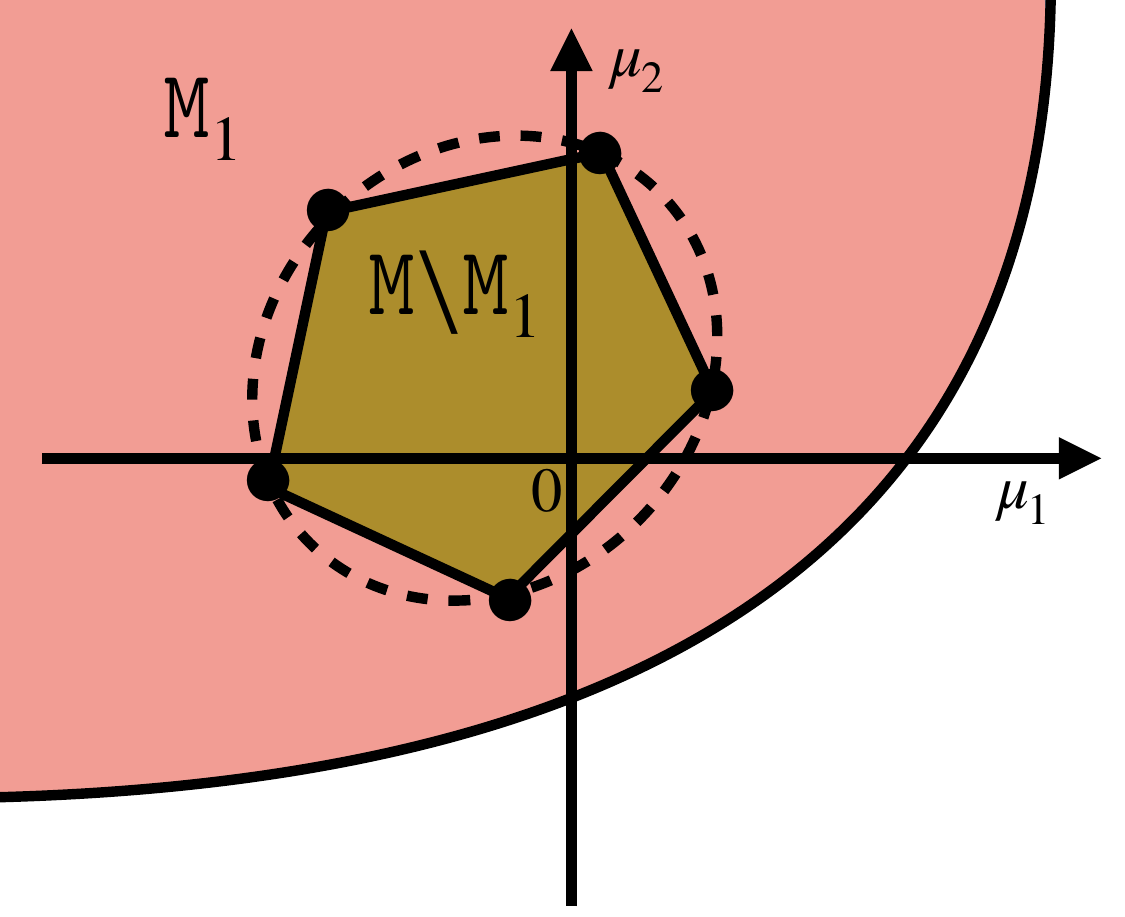}
        \caption{\label{fig:starA} 
        Surrounding, nice, $\mathtt{M}_1$ with a finite nice partition into convex sets. This figure is obtained by taking $P_0$ a Gamma distribution on $X$ and defining $Y= (Y_1,Y_2) = (\log X,X-c)$ for a constant $c > 1$. Then $\cE$ is a translated Gamma family with sufficient statistic $Y$ and mean-value space $\meanspace= \{(y_1,y_2): y_1 \in \reals, y_2 = e^{y_1} - c \}$   (unlike in Figure~\ref{fig:enter-label}, we have $\meanspace_1 \subset \meanspace$ here).}
	\end{minipage}\ \ \ \ 
	\begin{minipage}[t]{0.5\linewidth}
		\centering
		\includegraphics[width=2.2in]{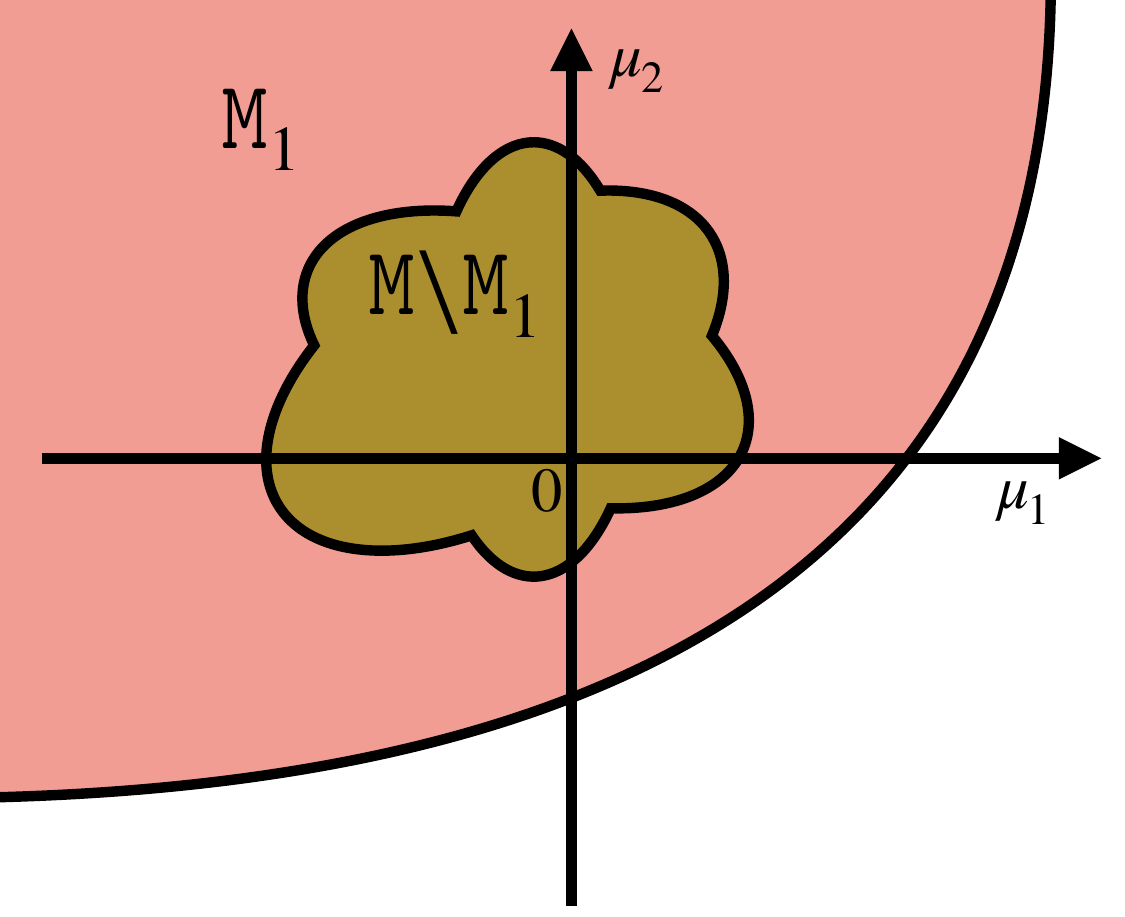}
         \caption{\label{fig:starB} Surrounding, nice $\mathtt{M}_1$ that cannot be partitioned into a finite number of convex sets (again we show the translated Gamma family).}
	\end{minipage}
\end{figure}
\commentout{
Our results will hold for surrounding $\Theta_1$ such that 
In practice it may be more likely that we want to bound the probability of a surrounding set $\meanspace_1$ in the mean-value space. Since $P_0(\mu \in \meanspace_1) = P_0(\theta \in \Theta_1)$ with $\Theta_1 = \theta(\meanspace_1)$, we may use the results to get such a bound. In principle it might happen that $\meanspace_1$, although itself $\nv$-star-shaped, whereas $\theta(\meanspace_1)$ is not. In that case, we can use the results below to get an upper bound on the $P_0$-probability of the $\nv$-star-shaped hull $\textsc{star}_\nv(\Theta_1)$, and uses this to further upper bound the $P_0$-probability of $\meanspace_1$, which cannot be larger.}

We now develop optimal e-variables for surrounding and regular $\meanspace_1$. 
The GROW criterion is still meaningful in this setting, and we discuss it in Section~\ref{sec:grow} below. Yet alternative, {\em relative\/} growth criteria are sometimes more meaningful in hypothesis testing \citep{ramdas2023savi,GrunwaldHK19} and one of these, minimax regret,  more directly leads to  corresponding CSC-type bounds. We consider these in Section~\ref{sec:regret} and~\ref{sec:inequality_theorem}. 
\subsection{GROW for $d=1$}\label{sec:grow}
We may again consider the GROW criterion for general $\cH_1$
that could be any set of distributions compatible with a given nice surrounding $\meanspace_1$ and $d \geq 1$, but this turns out to be surprisingly complicated in general. 
We only managed to find a simple characterization of $\Sgrow$ for the case $d=1$, $\meanspace_1 \subset \meanspace$ and $\cH_1 =  \cE_1= \{\bar{P}_{\mu}: \mu \in \meanspace_1\}$; that is, we are now testing $P_0$, a member of 1-dimensional exponential family $\cE$, against a subset $\cE_1$ of $\cE$ that is bounded away from $P_0$.

In the sequel we denote by $\pmv_{W}(X^n) := \int \pmv_{\mu}(X^n) dW(\mu)$ the Bayes marginal density corresponding to prior measure $W$.

\begin{theorem}\label{thm:main_new}
Let $P_0$ be a distribution for 1-dimensional $Y \subset \reals$, and suppose that $\meanspace_1$ is nice, i.e. $\meanspace_1^{\comp} = (\mu^{-}_1,\mu^+_1)$ is an open interval containing $\nv$ and contained in  the mean-value parameter space $\meanspace$ for the 1-dimensional exponential family generated by $P_0$. 
Then, among all distributions $W$ on the boundary $\bd(\meanspace_1^{\comp})= \{ \mu^-_1,\mu^+_1\}$, the minimum $D(\bar{P}_{W}|| P_{0})$
is achieved by a distribution $W^*$ that satisfies 
\begin{equation}\label{eq:jantje}
D(\bar{P}_{W^*}|| \bar{P}_{\mu_0}) = {\mathbb E}_{\bar{P}_{\mu^-_1}} \left[ 
\log \frac{\bar{p}_{W^*}(Y)}{p_{0}(Y)}
\right] = 
{\mathbb E}_{\bar{P}_{\mu^+_1}} \left[ \log \frac{\bar{p}_{W^*}(Y)}{p_{0}(Y)}
\right].
\end{equation}
The GROW e-variable relative to  $\cE_1 =\{ \bar{P}_{\mu}: \mu \in \meanspace_1\cap \meanspace\}$, denoted $S_{\grow}$, and the GROW e-variable relative to $\cE_1^{\bd} := \{\bar{P}_{\mu}: \mu \in \bd(\meanspace_1^{\comp})\}$, denoted $S_{\grow}^{\bd}$, are {\em both\/} given by:
\begin{equation}\label{eq:pietje}
S_{\textsc{grow}} = S^{\bd}_{\textsc{grow}}
= \frac{\bar{p}_{{W}^*}(Y)}{p_{0}(Y)},
\end{equation}
i.e., the support of the prior ${W}^*$ on $\meanspace_1$ minimizing the KL divergence $D(\bar{P}_{{W}^*} \| P_0)$ is fully concentrated on the boundary of $\meanspace_1^{\comp}$. 
\end{theorem}
\begin{proof}{\bf (of Theorem~\ref{thm:main_new}),}
Define, for $\mu \in \meanspace$ 
and $w^* \in [0,1]$,
\begin{align}
 \label{eq:brenda}
& f(\mu,w^*) = 
\mathbb{E}_{ \bar{P}_\mu} 
\left[\log \frac{(1-w^*) \bar{p}_{\mu^-_1}(Y)+  w^* \bar{p}_{\mu^+_1}(Y)}{p_{\mu_0}(Y)} \right]
\end{align}
and note that, for each $\mu$, it holds that  $f(\mu,w^*)$ is continuous in $w^*$.
%
Now consider $f(\mu^-_1,w^*) = - D(\bar{P}_{\mu^-_1} \| \bar{P}_{W^*}) + D(\bar{P}_{\mu^-_1} \| {P}_{0})$. Minus the first term, $D(\bar{P}_{\mu^-_1} \| \bar{P}_{W^*})$, 
is $0$ at $w^*=0$ and continuously monotone increasing in $w^*$, 
since KL divergence is nonnegative and strictly convex in its second argument. Therefore $f(\mu^-_1,w^*)$ is itself continuously monotone decreasing in $w^*$
with $f(\mu^-_1,0) = D(\bar{P}_{\mu^-_1} \| P_0) > 0$. Also, $f(\mu^-_1,1) = - D(\bar{P}_{\mu^-_1} \| \bar{P}_{\mu^+_1}) + D(\bar{P}_{\mu^-_1} \| P_{0}) < 0$, since KL divergence $D(\bar{P}_{\mu^-_1} \| \bar{P}_{\mu'})$ is strictly increasing in $\mu'$, for $\mu'> \mu^-_1$.


Analogously  $f(\mu^+_1,w^*)$ is continuously monotone increasing in $w^*$,
$f(\mu^+_1,1) = D(\bar{P}_{\mu^+_1} \| P_0) > 0$ and  $f(\mu^+_1,0) < 0$.

This shows that there exists $0 < w^{\circ}< 1 $ such that $f(\mu^-_1,w^{\circ}) =
f(\mu^+_1,w^{\circ})$. This implies that there exists a $W^*$ (with $W^*(\{\mu_1^+ \}) = w^{\circ}$) such that 
%
the rightmost equality in (\ref{eq:jantje}) holds. But this rightmost equality implies that for all $w' \in [0,1]$, 
$$
(1-w') {\mathbb E}_{\bar{P}_{\mu^-_1}} \left[ 
\log \frac{\bar{p}_{W^*}(Y)}{p_{0}(Y)}
\right] + w'
{\mathbb E}_{\bar{P}_{\mu^+_1}} \left[ \log \frac{\bar{p}_{W^*}(Y)}{p_{0}(Y)}
\right] = {\mathbb E}_{\bar{P}_{\mu^-_1}} \left[ 
\log \frac{\bar{p}_{W^*}(Y)}{p_{0}(Y)}
\right].
$$
Plugging in $w' = w^{\circ}$, we get the left equality in (\ref{eq:jantje}). 
Now (\ref{eq:jantje}) in turn gives that 
\begin{equation}\label{eq:barcelona}
\min_{\mu \in \bd(\meanspace_1^{\comp})} 
{\mathbb E}_{\bar{P}_{\mu}} \left[ 
\log \frac{\bar{p}_{W^*}(Y)}{p_0(Y)}
\right] = D(\bar{P}_{W^*}|| \bar{P}_{\mu_0}),
\end{equation}
whereas for any probability density $q$ for $\cY$,
\begin{equation}\label{eq:finalwork}
\min_{\mu \in \bd(\meanspace_1^{\comp})} 
{\mathbb E}_{\bar{P}_{\mu}} \left[ 
\log \frac{q(Y)}{p_0(Y)}
\right] 
\leq {\mathbb E}_{\bar{P}_{W^*}} \left[ 
\log \frac{q(Y)}{p_0(Y)}
\right]
\leq 
 {\mathbb E}_{\bar{P}_{W^*}} \left[ 
\log \frac{\bar{p}_{W^*}(Y)}{p_0(Y)}
\right]= 
D(\bar{P}_{W^*}|| \bar{P}_{\mu_0}),
\end{equation}
which together with (\ref{eq:barcelona}) shows that $S^{\bd}_{\grow} = \frac{\bar{p}_{{W}^*}(Y)}{p_{0}(Y)}$, since every well-defined e-variable can be written as $q(Y)/p_0(Y)$ for some probability density $q$. 
Also, similarly to
(\ref{eq:finalwork}), we have 
\begin{equation*}
\inf_{\mu \in \meanspace_1} 
{\mathbb E}_{\bar{P}_{\mu}} \left[ 
\log \frac{q(Y)}{p_0(Y)}
\right] 
\leq {\mathbb E}_{\bar{P}_{W^*}} \left[ 
\log \frac{q(Y)}{p_0(Y)}
\right] \leq 
D(\bar{P}_{W^*}|| \bar{P}_{\mu_0}),
\end{equation*}
so if we could show
\begin{equation}\label{eq:realfinalwork}
    \inf_{\mu \in \meanspace_1} 
{\mathbb E}_{\bar{P}_{\mu}} \left[ 
\log \frac{\bar{p}_{W^*}(Y)}{p_0(Y)}
\right]= 
D(\bar{P}_{W^*}|| \bar{P}_{\mu_0}),
\end{equation}
then the above two statements  would together also imply that $S_{\grow} = \frac{\bar{p}_{{W}^*}(Y)}{p_{0}(Y)}$, and we would be done. 
But (\ref{eq:realfinalwork}) follows  by (\ref{eq:jantje})  
\commentout{
This in turn implies that (a) for all $w\in [0,1]$ with $w \neq w^{\circ}$,
$\inf_{\mu \in \{\mu_1^-,\mu_1^+\}} f(\mu,w) < f(\mu_1^-,w^{\circ}) = f(\mu_1^+,w^{\circ})$, whereas (b) for both $\mu \in \{\mu_1^-,\mu_1^+\}$, $\sup_{w \in [0,1]} f(\mu,w) \geq f(\mu_1^-,w^{\circ}) = f(\mu_1^+,w^{\circ})$. (a) and (b) together 
show that $w^{\circ}$ achieves
$$
\max_{w \in [0,1]} \inf_{\mu \in \{\mu_1^-,\mu_1^+\}} f(\mu,w) = \inf_{\mu \in \{\mu_1^-,\mu_1^+\}} \max_{w \in [0,1]}  f(\mu,w)
$$
which shows that the minimum $D(\bar{P}_{W} \| P_0)$ over all $W$ is achieved by $W^*$, verifying the claim above (\ref{eq:jantje}), and that the
GROW e-variable relative to  $\bd(\meanspace_1^{\comp})$ is given by (\ref{eq:pietje}). 
It only remains to show that 
\begin{align}
\nonumber
\min\limits_{\mu \in {\meanspace}_1} f(\mu,w^{\circ}) 
=\min\limits_{\mu \in \bd({\meanspace}_1^{\comp})} f(\mu,w^{\circ}).
\end{align}
But this is implied by 
}
together with the following lemma, which thus completes the proof: 
\begin{lemma}\label{lem:brown} $f(\mu,w^{\circ})$
is increasing on $\{ \mu \in \meanspace_1: \mu \geq \mu^+ \}$ and decreasing on 
$\{ \mu \in \meanspace_1: \mu \leq \mu^- \}$.
\end{lemma}
\end{proof}
Lemma~\ref{lem:brown} is proved in  Section~\ref{sec:proof}. It follows from the fact that exponential families represent {\em variation reducing kernels} \citep{brown1981variation}, a notion in theoretical statistics that seems to have been largely forgotten, and that we recall in Section~\ref{sec:proof}. In that section we also explain why this result, even for $d=1$, is difficult to prove, which also explains why proving anything nonasymptotic for the case $d > 1$ is currently beyond our reach. 

\subsection{Alternative Optimality Criteria: minimax redundancy and regret}\label{sec:regret}
$\Sgrow$ is difficult to characterize when $d > 1$ and $\cH_1$ is surrounding; while it seems intuitive that, at least under some additional regularity conditions, it is still given by a likelihood ratio with a Bayes mixture concentrated on the boundary of $\meanspace_1$, we did not manage to prove this (we indicate what the difficulty is towards the end of  Section~\ref{sec:proof}); we can say more though for the special case that $\meanspace_1^{\comp}$ is a KL ball and $Y = n^{-1} \sum_{i=1}^n X_i$ as $n \rightarrow \infty$; see Section~\ref{sec:asymptotic_growth_rate}.

However, in e-variable practice we often deal with $\cH_1$ that can be partitioned into a family of subsets $\{\cH_{1,r}: r \in \cR\}$ such that, ideally, we would like to use the GROW e-variable relative to the $\cR$ that actually contains the alternative. This was called the {\em relative GROW\/} criterion by \cite{GrunwaldHK19} and it was used by, for example, \cite{TurnerLG24}. We thus face a collection of e-variables $\cS= \{S_{r,\grow} : r \in \cR\}$ where each $S_{r,\grow}$ is GROW for $\cH_{1,r}$. If an oracle told us beforehand  ``if the data come from $P \in \cH_1$, then in fact $P \in \cH_{1,r}$''
then we would want to use $S_{r,\grow}$. Not having access to such an oracle, we want to use an e-variable that loses the least e-power compared to $S_{r,\grow}$ for the `true' $r:= r(P)$ (i.e. such that $P \in \cH_{1,r}$), in the worst-case over all $r \in \cR$. 
This is akin to what in the information theory literature is called a {\em minimax redundancy\/} approach \citep{ClarkeB94,grunwald2007minimum}. This approach is still hard to analyze in general but is amenable to the asymptotic analysis we provide in the next section. 
In a minor variation of this idea, 
 used in an e-value context before by \cite{orabona2021tight,jang2023tighter}, we may consider the e-variable that loses the least e-power compared $S_{\breve{r},\grow}$ for the $\breve{r}$ that is optimal with hindsight for the data at hand, i.e. achieving $\max_{r \in \cR} S_{r,\grow}$; thus $\breve{r}$ can be thought of as a maximum likelihood estimator. This is akin to what information theorists call {\em minimax individual-sequence regret\/} \citep{grunwald2007minimum}. This final approach {\em can\/} be analyzed nonasymptotically and leads to an analogue of the CSC theorem. We now formalizeboth the redundancy and regret approaches.

Thus, suppose that $\meanspace_1$ is surrounding and {\em nice\/} as defined in the beginning of Section~\ref{sec:surrounding}, and let
$\{ \meanspace_{1,r}: r \in \cR\}$ be a partition of $\meanspace_{1}$ (see Figure~\ref{fig:starA}). We will restrict attention to {\em nice\/} partitions, i.e. partitions  such that 
\begin{equation}\label{eq:nicepart}
\text{for all $r \in \cR$,\  
$\meanspace_{1,r}$ is a closed convex subset of $\meanspace_1$ 
with $\meanspace_{1,r} \cap \bd(\meanspace_1^{\comp})\neq \emptyset$}
\end{equation}
Let $\cH_{1,r} := \{P \in \cH_1: {\bf E}_{Y \sim P}[Y] \in \meanspace_{1,r}\}$. 

By strict convexity of $D(\Pmv_{\mu} \| P_0)$ in $\mu$ and the nice-ness condition, we have that 
$\min_{\mu \in \meanspace_{1,r} \cap \meanspace} D(\Pmv_{\mu} \| P_0)$ exists and is achieved uniquely for a point $\mu^*(r)$ on the boundary $\bd(\meanspace_{1,r}) \cap \bd(\meanspace_1^{\comp})$. Every $r \in \cR$ is mapped to a point $f(r) \in \bd(\meanspace_1^{\comp})$ in this way, and the mapping $f: \cR \rightarrow \bd(\meanspace_1^{\comp})$ is injective since $\meanspace_{1,r} \cap \meanspace_{1,r'} = \emptyset$ for every $r \neq r'$ with $r, r' \in \cR$. Therefore we will simply {\em identify\/} $\cR$ with a subset of $\bd(\meanspace_1^{\comp})$, such that $f(\mu) = \mu$ for all $\mu \in \cR$.  For $P \in \cH_1$ with ${\mathbb E}_P[Y] = \mu$ (so $\mu \in \meanspace_1$), we will now define  $r(P)$ to be the $r \in \cR$ such that $P \in \cH_{1,r}$, i.e. such that $\mu \in \meanspace_{1,r}$. Note we can think of $r(P)$ either as an index of sub-hypothesis $\cH_{1,r}$ or as a special boundary point of the space of mean-values $\meanspace_{1,r}$. 

If we were to test $\cH_0$ vs. $\cH_{1,r}$ for given $r$, then we would still like to use the GROW e-variable $S_{r,\grow} = \bar{p}_{r}/p_0$.
In reality we do not know $r$, but we aim for an e-value that loses as little evidence as possible compared to $S_{r,\grow}$, in the worst-case over all $r$. 
Formally, we seek to find e-variable $S= q(Y)/p_0(Y)$, where  $q$ achieves 
\begin{align}\label{eq:redundancy}
 &   \sup_q \inf_{P 
    \in \cH_1}\ {\bf E}_{Y \sim P} \left[ \log \frac{q(Y)}{p_0(Y)} - \log \frac{\pmv_{r(P)}(Y)}{p_0(Y)} \ \right]= 
     \sup_q \
     \inf_{P 
    \in \cH_1}\ {\bf E}_{Y \sim P} \left[
     \log \frac{q(Y)}{\pmv_{r(P)}(y)} \ \right] \nonumber \\
     & = - \mmred(\cH_1) \text{\ where\ } \mmred(\cH_1) := \inf_q \sup_{P \in \cH_1}
     \left( D(P \|Q) - D(P \| \bar{P}_{r(P)}) \right).  \end{align}
where the supremum is over all probability densities on $Y$ and $r(P)$ is again the unique $r \in \cR \subset \meanspace_1$ such that $P \in \cH_{1,r}$. $\mmred(\cH)$  is easily shown to be nonnegative for any $\cH$, and both equations in (\ref{eq:redundancy}) are immediate. From the rightmost expression, information theorists will recognize the $q$ as minimizing the maximum {\em redundancy\/} \citep{CoverT91,ClarkeB94,takeuchi2013asymptotically}: the worst-case additional mean number of bits needed to encode the data  by an encoder who only knows that $P \in \cH_1$ compared to an encoder with the additional knowledge that $P \in \cH_{1,r}$.  

As said, it is easier to analyze a slight variation of this approach which makes at least as much sense:
rather than comparing ourselves to the inherently unknowable $r(P)$, we may consider the actually observed data $Y=y$ and compare ourselves to (i.e. try to obtain as much evidence against $P_0$ as possible compared to) the $\breve{r}(y)$ we would like to have used with hindsight, after seeing and in light of data $y$; and rather than optimizing an expectation under an imagined distribution which we will never fully identify anyway, we will optimize in the worst-case over all data.
The setup works for general functions $\breve{r}$, indicating what $r\in \cR$ we would have liked to use with hindsight; further below we discuss intuitive choices. 
We note that $\breve{r}$ maps data $\cY$ to a point in $\cR \subset \bd(\meanspace_1^{\comp}) \subset \meanspace_1$, so we can think of $\breve{r}$ as an {\em estimator\/} of parameter $\mu$; however, the estimator is restricted to a small subset of $\meanspace$, namely the set $\cR$.

Thus, 
we now seek to find e-variable $\Srel= q(Y)/p_0(Y)$, where  $q$ now achieves 
\begin{align}\label{eq:shtarkov}
&     \sup_q \inf_{y \in \cY}\ \left( \log \frac{q(y)}{p_0(y)} - \log \frac{\pmv_{\breve{r}(y))}(y)}{p_0(y)} \ \right)= 
     \sup_q \inf_{y \in \cY}\ \left( \log \frac{q(y)}{\pmv_{\breve{r}(y)}(y)} \ \right) \nonumber\\
     & = - \mmreg(\breve{r}) \text{\ where\ } \mmreg(\breve{r}) := \inf_q \sup_{y \in \cY}\ \left( - \log \frac{q(y)}{\pmv_{\breve{r}(y)}(y)} \ \right),
\end{align}
where again the supremum is over all probability densities that can be defined on $Y$, the quantity $\mmreg(\breve{r})$ is easily seen to be nonnegative regardless of how $\breve{r}$ is defined, and both equalities in (\ref{eq:shtarkov}) are immediate. 
From the rightmost expression, information theorists will recognize the optimizing $q$ as the $q$ minimizing {\em individual sequence regret\/}: it minimizes the `regret' in terms of the number of bits needed to encode the data, in the worst-case over all sequences, compared to somebody who has seen $\breve{r}(y)$ in advance; the word `regret' is also meaningful in our setting --- the aim is to minimize regret in the sense of loosing as little evidence as possible compared to the largest attainable  e-value (evidence) with hindsight. 
In information theory, neither the minimax redundancy nor the minimax individual sequence regret is considered inherently superior or more natural, and (as we shall also see in our context, in the next section), both quantities often behave similarly.

Indeed, (\ref{eq:shtarkov}) being a variation of a standard problem within information theory and sequential prediction with the logarithmic loss, it is well-known 
\citep{BarronRY98,CesaBianchiL06,grunwald2007minimum,GrunwaldM19} that the solution for $q$ is uniquely achieved by the following variation of the {\em Shtarkov distribution}, a notion going back to \cite{Shtarkov87}: \begin{equation}\label{eq:shtarkovb}
    q_{\shtarkov,\breve{r}}(y) = \frac{\pmv_{\breve{r}(y)}(y)}
    {\int_y \pmv_{\breve{r}(y))}(y)\nu(dy)}
    \ \text{so}\ \Srel = \frac{q_{\shtarkov}(Y)}{p_0(Y)}.
\end{equation}
We then get that $
\left( \log \frac{q(y)}{\pmv_{\breve{r}(y)}(y)} \ \right) = \log \int_y \pmv_{\breve{r}(y)}(y)\rho(dy) := 
\mmreg(\breve{r})
$
independently of $y$, where $\mmreg(\breve{r})$ is the {\em maximin regret\/} (called `minimax regret' originally, since in data compression the rightmost expression, without the minus sign, is the relevant one). 

The most straightforward choice is to take $\breve{r}(y):= \hat{r}(y)$ the maximum likelihood estimator within $\cR$, 
achieving 
\begin{equation}\label{eq:mle}
    \max_{r \in \cR} \bar{p}_{r}(y)
\end{equation}
for the given $y$,
since then $q_{\shtarkov}$ has minimal overhead compared to the  $S_{\grow,r}$ that is largest with hindsight, i.e. that provides the most evidence with hindsight --- thus providing additional justification for the terminology `regret' in this special case, which was also Shtarkov's (\citeyear{Shtarkov87}) original focus. If $\breve{r}$ is set to $\hat{r}$, then  $Q_{\shtarkov}$ is also known as the {\em normalized maximum likelihood (NML)\/} distribution relative to the set $\cR \subset \bd(\meanspace_1^{\comp})$ (not relative to the full exponential family $\cE$!).

In the ensuing results we will mostly be interested in the case that $\cR$ is either finite (as in Figure~\ref{fig:starA}), or that it is in 1-to-1 correspondence with $\bd(\meanspace^{\comp}_1)$ (as in Figure~\ref{fig:starB}). In both cases, the maximum in (\ref{eq:mle}) is achieved, although it may be achieved for more than one $r$; in that case, we set $\hat{r}$ to be the largest $r$ achieving (\ref{eq:mle}) in lexicographical ordering, making $\hat{r}$ well-defined in all cases.

While the upcoming analogue of the CSC theorem will mention the MLE $\hat{r}$, it turns out that in the proof, and in the detailed theorem statement, we also need to refer to an estimator  $\breve{r}$  that may differ from the MLE. The reason is that, intriguingly, in general we may have that for some $r \in \cR$ and  $y \in \meanspace_{1,r}$ we may have that $\hat{r}(y) \neq r$,  which complicates the picture. To this end, we formulated the minimax regret approach for general $\breve{r}: \conv(\cY) \rightarrow {\cal R}$, and not just the MLE, as was done earlier e.g. by \cite{GrunwaldM19}.  

\begin{example}{\bf [Gaussian Example]}
{\rm     We use the simple 2-dimensional Gaussian case  to gain intuition.
    Thus, we let $\cY = \conv(\cY) = \reals^2$, and $P_0$ a 2-dimensional Gaussian distribution on $Y$, with mean $0$ (i.e., $(0,0)$) and  $\Sigma$ a positive definite $2 \times 2$ covariance matrix. Then $\meanspace = \reals^2$ and 
    $\cE$ is a 2-dimensional Gaussian location family. For now, take $\Sigma$ to be the identity matrix. Then $D(P_0 \| P_{\mu}) = (1/2) \| \mu \|_2^2$ is simply the squared norm of $\mu$, facilitating the reasoning.   
    A simple case of a convex $\meanspace_1$ is the translated half-space $\meanspace_1 = \{ (y_1,y_2): y_1 \geq a\}$ for constant $a > 0$. The point  $\mu \in \meanspace_1$ minimizing KL divergence to $P_0$ must then clearly be $(a,0)$. Therefore, if $\cH_1$ is any set of distributions with means in $\meanspace_1$ and containing $\cE_1^{\bd} = \{ \bar{P}_{\mu}: \mu \in \bd(\meanspace_1^{\comp})\}$, we have by Proposition~\ref{prop:exp} that 
    $S_{\grow} = \bar{p}_{(a,0)}(Y)/\bar{p}_{(0,0)}(Y)$.
    We see that even if we have convex $\cH_1$, the minimax individual sequence regret approach leads to a different e-variable, if we carve up $\cH_1$ into $\{\cH_{1,r}: r \in \cR\}$ for $\cR$ with more than one element. For example, we can take $\cR= \{(a,\mu_2): \mu_2 \in \reals\}$ be a vertical line and let $\meanspace_{1,(a,\mu_2)}$ be the subset of $\meanspace_1$ on the line connecting $(0,0)$ with $(a,\mu_2)$. Then, with $\hat{r}$ the MLE as in (\ref{eq:mle}), we get that $\hat{r}(y)$ is the point where the line $\cR$ intersects with the line connecting $(0,0)$ and $y$. Since now $\hat{r}(y)$, and hence the sub-hypothesis $\cH_{1,\hat{r}(y)}$ we want to be almost GROW against, changes with $y$, we get a solution $\Srel$ in (\ref{eq:shtarkovb}) that differs from $S_{\grow}$.

In the case of convex $\cH_1$, whether to use the absolute or relative GROW e-variable may depend on the situation (see \cite{GrunwaldHK19} for a motivation of when absolute or relative is more appropriate). In the case of nonconvex $\cH_1$ with $d > 1$, we simply do not know how to characterize the absolute GROW and we have to resort to the relative GROW. }
\end{example}

\subsection{CSC Theorem for surrounding $\cH_1$}\label{sec:inequality_theorem}
For given estimator $\breve{r}$ and probability density $q$ on $Y$, define $\reg(q,\breve{r})$ (`regret') as
\begin{equation}\label{eq:regret}
\reg(q,\breve{r},y) := \log  \left( 
\frac{\pmv_{\breve{r}(y)}(y)}{q(y)}
\right) \ ; \ \mreg(q,\breve{r}) := \sup_{y: y \in \meanspace_1}
\log  \left( 
\frac{\pmv_{\breve{r}(y)}(y)}{q(y)}
\right).
\end{equation}
Whereas above we discussed an MLE $\breve{r}$, i.e. $\breve{r}=\hat{r}$, we now require it to be self-consistent, i.e. we set it to be any function of $y$ such that 
for all $y \in \meanspace_1$, we have $y \in \meanspace_{1,\breve{r}(y)}$. 
The value of $r(y)$ for $y \in \conv(\cY) \setminus \meanspace_1$ will not affect the result below.


\begin{theorem}{\rm \bf [CSC Theorem for surrounding $\cH_1$]}
    \label{thm:nml}
    Suppose that $\meanspace_1$ is nice and let $ \{ \meanspace_{1,r}:  r \in \cR \}$ be any nice partition of $\meanspace_1$ as in (\ref{eq:nicepart}), with $\breve{r}$ any self-consistent estimator as above. Let $q$ be an arbitrary probability density function. Then: 
    $$
{\mathbb E}_{P_{0}}\left[ {{\bf 1}_{Y \in \meanspace_1}}\cdot e^{D(\Pmv_{\breve{r}(Y))} \| P_0) - \reg(q,\breve{r},Y)}\right] \leq 1.
$$  
so that in particular,  with  $\underline{D}:=  \inf_{\mu \in \bd(\meanspace_1^{\comp})} D(P_{\mu}\| P_{0})$,
\begin{equation}\label{eq:mlesharpbound}
P_{0}\left( Y\in \meanspace_1
\right)\leq  e^{\mreg(q,\breve{r}) - \underline{D}} \overset{(*)}{=}
e^{\mmreg(\breve{r})- \underline{D}}
\overset{(**)}{\leq} e^{\mmreg(\hat{r})- \underline{D}}
,
\end{equation}
where $(*)$ holds if we take $q= q_{\shtarkov,\breve{r}}$,
and $(**)$ holds if the MLE estimator $\hat{r}$ is well-defined. 
\end{theorem}
\begin{proof}
Folllowing precisely analogous steps as in the proof of (\ref{eq:subset}) as based on (\ref{eq:meanimp}) within the proof of Theorem~\ref{thm:old}, we obtain, for all $y \in \conv(\cY)$,
\begin{equation}\label{eq:subsetb}
y \in \meanspace_1 \Rightarrow \frac{\pmv_{\breve{r}(y)}(y)}{p_0(y)} \geq  e^{D(\Pmv_{\breve{r}(y)} \| P_0)}.
\end{equation}
Then (\ref{eq:subsetb}) gives, using definition  (\ref{eq:regret}):
\begin{align*}
 &  {\mathbb E}_{P_{0}}\left[ {{\bf 1}_{Y \in \meanspace_1}}\cdot {e^{D(\Pmv_{\breve{r}(Y)} \| P_0) - \reg(q,\breve{r},Y)}}\right] \leq 
  {\mathbb E}_{P_{0}}\left[ {{\bf 1}_{Y \in \meanspace_1}}\cdot 
 \frac{\pmv_{\breve{r}(Y)}(Y)}{p_0(Y)}  \cdot 
  {e^{- \reg(q,\breve{r},Y)}}\right] \leq \\ &
 {\mathbb E}_{P_{0}}\left[ {{\bf 1}_{Y \in \meanspace_1}}\cdot 
 \frac{q(Y)}{p_0(Y)}  \cdot e^{ \reg(q,\breve{r},Y)} \cdot 
  {e^{-  \reg(q,\breve{r},Y)}}\right] \leq 
   {\mathbb E}_{P_{0}}\left[ 
 \frac{q(Y)}{p_0(Y)} \right]= 1,
\end{align*}
which proves the first statement in the theorem. The second statement is then immediate. 
\commentout{Combining the two displays, we have: 
\begin{align*}
P_{\theta_0}\left( \hat\theta \in \Theta_1
\right) 
&\leq  
P_{\theta_0}\left( {\bf E}_{\hat\theta} \left[ \log 
\frac{p_{g(\hat\theta)}(Y)}{p_{\theta_0}(Y)} \right] \geq   
D(P_{g(\hat\theta)} \| P_{\theta_0})
\right) \\
&\leq P_{\theta_0}\left( {\bf E}_{\hat\theta} \left[ \log 
\frac{p_{g(\hat\theta)}(Y)}{p_{\theta_0}(Y)} \right] \geq  \underline{D}
\right) \\ 
&\stackrel{(a)}{=} P_{\theta_0} \left(   \log 
\frac{p_{g(\hat\theta)}({Y})}{p_{\theta_0}({Y})}  \geq  \underline{D}
\right)
\end{align*}
where $(a)$ holds by (\ref{eq:robustnessb}).

Then proceeding from above, 
\begin{align*}
P_{\theta_0}\left( \hat\theta \in \Theta_1
\right)
&\leq P_{\theta_0}\left(   \log 
\frac{q({Y})}{p_{\theta_0}({Y})} + \reg(q)  \geq \underline{D}
\right) \\ 
&= P_{\theta_0}\left(   
\frac{q({Y})}{p_{\theta_0}({Y})}  \geq  e^{-\reg(q) + \underline{D}} 
\right) \\ 
&\stackrel{(b)}{\leq} {\bf E}_{\theta_0} \left[ \frac{q({Y})}{p_{\theta_0}({Y})}\right] \cdot e^{\reg(q) -\underline{D}}= e^{\reg(q) - \underline{D}}
\end{align*}
where $(b)$ holds by Markov's inequality.

We have thus shown, in analogy to Theorem~\ref{thm:old}, Part 2:
}
\end{proof}
\paragraph{Analyzing the CSC Result --- Different Partitions $\{\meanspace_{1,r}: r \in \cR\}$}
The next question is how to cleverly partition any given nonconvex $\meanspace_1$ so as to get a good bound when applying Theorem~\ref{thm:nml}.
We first note that, for any given $\meanspace_1$,  the final bound (\ref{eq:mlesharpbound}) 
does not worsen if we enlarge $\meanspace_1$, as long as 
 $\underline{D}=  \inf_{\mu \in \bd(\meanspace_1^{\comp})} D(P_{\mu}\| P_{0})$ stays the same. Thus, we still get the same bound if we shrink the  complement $\meanspace_1^{\comp}$ to the $\underline{D}$-KL ball $\{ \mu: D(\Pmv_{\mu} \|P_0) < \underline{D} \}$, without making the bound looser. We will therefore, from now on, simply assume that we are in the situation in which $\meanspace_1^{\comp}$ is the $\underline{D}$-KL ball. 
 We know that such a KL ball is convex \citep{Brown86} --- with such a convex  $\meanspace_1^{\comp}$ we are thus in the `dual case', as it were, to the case of convex $\meanspace_1$ which we discussed in Section~\ref{sec:convexm1}.  

There are now basically two approaches to apply the CSC Theorem~\ref{thm:nml} that suggest themselves. In the first approach, we first determine a larger $\meanspace'_1$ 
(hence smaller $\meanspace^{'\comp}_1$) that contains $\meanspace_1$, such that $\meanspace'_1$ can be partitioned into a finite number $|\cR'|$ of convex subsets, $\{ \meanspace'_{1,r}: r \in \cR \}$, and then we apply Theorem~\ref{thm:nml} to $\meanspace'_1$ and $\cR'$.
We could, for example, take $\meanspace_1^{'\comp}$ be a convex polytope with a finite number of corners, all touching $\bd(\meanspace_1^{\comp})$. Such a situation is depicted in Figure~\ref{fig:starA}, if we interpret the dashed curve the boundary of a KL ball and the $\meanspace^{\comp}_1= \meanspace \setminus \meanspace_1$ in the Figure as the polytope $\meanspace_1^{'\comp}$. 
\commentout{
\begin{example}{\bf [Gaussian Example, Continued]}
    {\rm Consider the 2-dimensional Gaussian case, with identity covariance matrix. 
    Then the $\underline{D}$-KL ball coincides with the Euclidean ball of radius $\underline{D}^{1/2}$ around $0$ and $\bd(\meanspace_1^{\comp})$ is a circle. In the first approach, we could take any $\meanspace'_1$ such that $\meanspace_1^{'\comp}$ is a regular polygon with $k \geq 3$ corners that is circumscribed by the circle, i.e. it touches the circle at its corners. Then we would take an $\cR$ with $k$ elements, one for each edge. We set each $r \in \cR$ equal to a midpoint of the corresponding edge, set $\meanspace_{r}$ equal to the cone from $0$ with sides going through the endpoints of this edge, and $\meanspace_{1,r}:= \meanspace_1 \cap \meanspace_r$. Then, for any $P \in \cH_1$ with ${\mathbb E}_P[Y]= \mu$, we have $r(P)$ is equal to the point in the cone containing $\mu$ that minimizes both Euclidean distance and KL divergence to $0$.
    }
\end{example}}
In the second approach, we set $\cR= \bd(\meanspace_1^{\comp})$, making it a manifold in $\reals^d$, and set $$
\meanspace_{1,r} := \{ \mu \in \meanspace_1: \alpha \mu = r \text{\ for some\ } \alpha > 0 \},$$
i.e. the set of points in $\meanspace_1$ on the ray starting at $0$ and going through $r$. Then we have for all $y \in \bd(\meanspace_1^{\comp})$ that  $r(y) = y$.
We may think of this second approach as a limiting case of the first one, when we let the number of corners of the polytope go to infinity. In the next section we show that, if we apply the CSC Theorem~\ref{thm:nml}, in our main case of interest, with $Y=n^{-1} \sum_{i=1}^n X_i$, and $\meanspace_1^{\comp}$ a KL Ball, then for large $n$, the second, `continuous' approach always leads to better bounds than the first.

\section{Asymptotic expression of growth rate and regret}\label{sec:asymptotic_growth_rate}
While the exact sizes of $\mmred(\cH_1)$ and $\mmreg(\breve{r})$  are   hard to determine, for the case of nice $\meanspace_1^{\comp}$ and our central case of interest, with $Y= n^{-1} \sum X_i$, we can use existing results to obtain relatively sharp asymptotic (in $n$) approximations of  $\mmred(\cH_1)$ and $\mmreg(\breve{r})$'s  upper bound $\mmreg(\hat{r})$. We now derive these approximations and show how they, in turn, lead to an approximation to $\grow$ if moreover $\meanspace_1^{\comp}$ is a KL ball, which as explained above, is also the case of central interest. 

Thus, we now assume that $Y := n^{-1} \sum X_i$, with $X, X_1, X_2, \ldots$ i.i.d. $\sim P'_0$, where $P'_0$ is a distribution on $X$, inducing distribution $P_0 \equiv P_0^{[Y]}$ for $Y$.  Like $P_0$, we have that 
$P'_0$ also generates an exponential family, now with sufficient statistic $X$ and densities
\begin{equation}\label{eq:trump}
p'_{\theta}(x) \propto e^{\theta^{\top}x } p'_0(x)
\end{equation}
extended to $n$ outcomes by assuming independence. 
Now ${\mathbb E}_{P'_{\theta}} [Y]= 
n^{-1} {\mathbb E}_{P'_\theta} [\sum_{i=1}^n X_i] = 
{\mathbb E}_{P'_{\theta}} [X]
$
so that the mapping $\mu(\theta)$ from canonical to mean-value parameter is the same for both the family (\ref{eq:trump}) and the original family $\{P_{\theta}(Y): \theta \in \Theta \}$, and the  likelihood ratio of any member $P'_{\theta}$ of the family (\ref{eq:trump}) to $P'_0$ on $n$ outcomes is given by, with $\mu = \mu(\theta)$, 
$$
\prod_{i=1}^n \frac{\bar{p}'_{\mu}(X_i)}{p'_0(X_i)} = \frac{p'_{\theta}(X^n)}{p'_0(X^n)}
= \frac{p_{\theta} (Y)}{p_0(Y)} = \frac{\bar{p}_{\mu}(Y)}{p_0(Y)},
$$
which in turn implies that 
$D(\bar{P}_{\mu} \| P_0) = n D(\bar{P}'_{\mu} \| P'_0 )$, where $ D(\bar{P}'_{\mu} \| P'_0 )$ is an expression that does not change with $n$. 
This means that if we keep $\meanspace_1$ constant, the $\underline{D}= \inf_{\mu \in \meanspace_1} D(\bar{P}_{\mu} \| P_0)$ in (\ref{eq:mlesharpbound}) increases linearly in $n$. 
To avoid confusion here, it is useful to make explicit the dependency of $\mmreg$ and $\underline{D}$ on $n$ in Theorem~\ref{thm:nml}, by writing it as $\mmreg_n$ and $\underline{D}_n$: we can then restate the bound (\ref{eq:mlesharpbound}) in the theorem as
\begin{equation}\label{eq:mlesharpboundB}
P_{0}\left( Y\in \meanspace_1
\right)\leq  
e^{\mmreg_n(\hat{r})- \underline{D}_n}=
e^{\mmreg_n(\hat{r})- n \underline{D}_1}.
\end{equation}
Thus, as $n$ increases, if we keep $\meanspace_1$ fixed, then the quantity $\underline{D}_n$ in the bound increases linearly in $n$. On the other hand, we will now show that, for sufficiently smooth boundaries of $\meanspace_1$, we have that $\mmreg_n(\hat{r})$ only increases logarithmically in $n$, making the strength of  bound (\ref{eq:mlesharpbound}) still grow exponentially  in $n$. The result is a direct corollary of a result by \cite{takeuchi2013asymptotically}. 
\begin{theorem}\label{thm:TakeuchiB}{\bf 
[Corollary of Theorem 2 of \cite{takeuchi2013asymptotically}; see also \cite{takeuchi1997asymptotically}]}
Consider the setting of surrounding $\meanspace_1$ and suppose that $\meanspace_1$ is {\em nice}, as defined in the beginning of Section~\ref{sec:surrounding},
and that there exists a bijective function $\phi: \cU \rightarrow \bd(\meanspace_1^{\comp})$ so that $\cU$ is a subset of $\reals^{d-1}$ with open interior, $\phi$ has at least four derivatives and these are bounded on $\cU$.
Then $\hat{r}$ is well-defined and there exists $C > 0$ such that for all  $n$,
$$\mmreg_n(\hat{r}) \leq  \frac{d-1}{2} \log n + C.$$
\end{theorem}
We also have a bound on the minimax regret for the case that $\cH_1$ contains $\cE_1^{\bd}$, the subset of exponential family $\cE$ restricted to the boundary $\bd(\meanspace_1^{\comp})$: 
\begin{theorem}\label{thm:ClarkeB}{\bf 
[Corollary of Theorem 1 of \cite{ClarkeB94}]}
Consider the setting and conditions of the previous theorem.
Let $\cE_1^{\bd} := \{ \bar{P}_{\mu}: \mu \in \bd(\meanspace_1^{\comp}) \cap \meanspace \}$. 
Then there exists $C' < 0$ such that for all  $n$,
$$\mmred_n(\cE_1^{\bd}) = \inf_q \sup_{ \mu \in \bd(\meanspace_1^{\comp})}
      D(\bar{P}_{\mu} \|Q)    \geq  \frac{d-1}{2} \log n + C' .$$
\end{theorem}
\begin{quote}
To see how this result follows from Clarke and Barron's, note that for this we need to verify the regularity Conditions 1--3 in Section 2 of their paper. the assumption of niceness implies that  $\bd(\meanspace_1^{\comp})$ is contained in a compact subset $\meanspace'$ of the interior of $\meanspace$. Then also the corresponding natural parameters are contained in a compact subset $\Theta'$ of the interior of $\Theta$. 
Condition 1 of their paper is immediately verified for parameters in the natural parameterization $\Theta$. Since the functions $\theta: \meanspace \rightarrow \Theta$ and $\phi: \cU \rightarrow \meanspace$ and their first and second derivatives are themselves bounded on $\cU$ and $\meanspace'$, Condition 1 of their paper is verified in terms of the relevant parameterization $\bar{P}_{\phi(u)}$ as well. Moreover,  for all $k \in \naturals$, all partial derivatives of form $\partial^k/(\partial u_{j_1} \ldots \partial u_{j_k}) \int \bar{p}_{\phi(u)}(y) d \rho(y)$, with $j_1, \ldots, j_k \in \{1, \ldots, d-1\}$, can be calculated by exchanging differentiation and integration (this follows from  \citep[Theorem 2.2]{Brown86}. Since we already established \cite{ClarkeB94}'s Condition 1, this implies that their Condition 2 also holds, as they explain underneath their Condition 2. Their Condition 3 is immediate. 
\end{quote}
Now, for any $\cH_1$ 
that contains $\cE_1^{\bd}$, Theorem~\ref{thm:ClarkeB} implies that 
\begin{equation}\label{eq:eleven}
\mmred_n(\cH_1) \geq \mmred_n(\cE_1^{\bd}) \geq  \frac{d-1}{2} \log n + C',
\end{equation}
and it is also immediate, by definition of $\hat{r}$, that
\begin{equation}\label{eq:twelve}
\mmreg_n(\hat{r}) \geq \mmred_n(\cH_1). 
\end{equation}
Together with (\ref{eq:eleven}) and (\ref{eq:twelve}), Theorem~\ref{thm:TakeuchiB} above now gives that, under the assumption of both theorems,
\begin{align}\label{eq:dminusone}
\mmred_n(\cH_1) = 
\mmreg_n(\hat{r}) + O(1) = 
\frac{d-1}{2} \log n + O(1).
\end{align}

\paragraph{\bf How to Partition $\meanspace_1$, Continued}
We now restrict to the case that $\meanspace_1^{\comp}$ is a $\underline{D}_1$-KL ball. At the end of previous section we explained why this is the major case of interest. 
As above, we want to keep $\meanspace_1^{\comp}$ fixed as $n$ increases, i.e. the set of parameters that stays in  $\meanspace_1$ does not change with $n$. This means that, when viewed as a KL ball of distributions on $X$, the radius of the ball remains constant with $n$, but when viewed as a KL ball of distributions on $Y$, the radius of the ball does need to scale linearly with $n$, i.e. we set:
\begin{equation}\label{eq:klball}
\meanspace_1^{\comp} =  \{ \mu: D(\bar{P}'_{\mu} \| P'_0) <  \underline{D}_1 \} = \{ \mu: D(\bar{P}_{\mu} \| P_0) <   n \underline{D}_1 \}.\end{equation}
We now return to the two approaches to applying the CSC Theorem~\ref{thm:nml} which we discussed at the end of the previous section: one based on a finite $\cR$ defining a polytope, one based on a `continuous' $\cR = \bd(\meanspace_1^{\comp})$.
It turns out that if $\meanspace_1$ is defined in terms of a KL ball (\ref{eq:klball}), then, for large $n$, it is always better to take the continuous $\cR$ approach. To see this, 
suppose that, in the polytope approach, we take a polytope $\meanspace_1^{'\comp}$ with $k$ corners (e.g., $k=5$ in Figure~\ref{fig:starA}); let $\breve{r}_{k}$ be the corresponding estimator and $\hat{r}_{k}$ be the corresponding MLE in $\cR$ and $\underline{D}'_{1,k} < \underline{D}_1$ (Figure~\ref{fig:starA}) indicates why the inequality holds) be the minimum KL divergence we then obtain in (\ref{eq:mlesharpbound}) when replacing $\meanspace_1^{\comp}$ by $\meanspace_1^{'\comp}$, applied for $n=1$; similarly we let $\breve{r}_{\infty}$ be the corresponding estimator for the second, continuous approach and $\hat{r}_{\infty}$ be the corresponding MLE in $\cR = \bd(\meanspace_1^{\comp})$ and $\underline{D}'_{1,\infty} = \underline{D}_1$ be the corresponding minimum KL divergence appearing in the bound.  Then the rightmost bounds in (\ref{eq:mlesharpbound}) will, respectively, look like
$$
e^{\mmreg_n(\breve{r}_k) -  n \underline{D}'_{1,k}} \leq
e^{\mmreg_n(\hat{r}_k) - n \underline{D}'_{1,k}}  \text{\ vs.\ }
e^{\mmreg_n(\breve{r}_{\infty}) - n \underline{D}_{1}} \leq
e^{\mmreg_n(\hat{r}_{\infty}) - n \underline{D}_1}.
$$
Since $\mmreg_n(\breve{r}_k)$ is always nonnegative, no matter the definition of $\breve{r}$ and the value of $k$ and $n$, whereas $\mmreg_n(\hat{r}_{\infty})$ is logarithmic, 
and $\underline{D}'_{1,k} < \underline{D}_1$ for all finite $k$, 
the continuous approach based on $k=\infty$ provides bounds that are  eventually exponentially better in $n$ compared to the bound based on any finite $k$. 

\paragraph{An Asymptotic GROW Result} In the KL ball setting, we can also say something about the asymptotic size of the worst-case optimal growth rate:
\begin{proposition}\label{prop:scaling}
In the KL ball setting above, let $\cH_1$ be any set of distributions that contains $\cE_1^{\bd}$ and with means contained in $\meanspace_1$, i.e. $\{{\mathbb E}_{P}[Y]: P \in \cH_1 \} \subset \meanspace_1$. Then the growth rate $\grow_n$ at sample size $n$ is given by 
    $$
    \grow_n = n \underline{D}  - \frac{d-1}{2} \log n + O(1).
    $$
\end{proposition}
\begin{proof}
    We have, with the supremum being taken over all probability densities $q$ for $Y$,
    \begin{align*}
        & \sup_q \inf_{P \in \cH_1} {\mathbb E}_{P}\left[
        \log \frac{q(Y)}{p_0(Y)}\right] \leq 
        \sup_q \inf_{\mu \in \bd(\meanspace_1^{\comp})} {\mathbb E}_{\bar{P}_{\mu}}\left[
        \log \frac{q(Y)}{p_0(Y)}\right] = \\
       &  \sup_q \inf_{\mu \in \bd(\meanspace_1^{\comp})} \left( \; {\mathbb E}_{\bar{P}_{\mu}}\left[
        \log \frac{q(Y)}{p_0(Y)} - \log \frac{\bar{p}_{\mu}(Y)}{p_0(Y)}\right] + {\mathbb E}_{\bar{P}_{\mu}}\left[
        \log \frac{\bar{p}_{\mu}(Y)}{p_0(Y)}\right] \; \right) =
        \\ & \sup_q \inf_{\mu \in \bd(\meanspace_1^{\comp})} {\mathbb E}_{\bar{P}_{\mu}}\left[
        \log \frac{q(Y)}{p_0(Y)} - \log \frac{\bar{p}_{\mu}(Y)}{p_0(Y)}\right] + n \underline{D}_1 = 
        - \frac{d-1}{2} \log n + O(1) + n \underline{D}_1.
    \end{align*}
    where the last equation follows from (\ref{eq:dminusone}).
    On the other hand, with $q_{\shtarkov,\hat{r}}$ defined as in (\ref{eq:shtarkovb}), we also have:
     \begin{align*}
        & \sup_q \inf_{P \in \cH_1} {\mathbb E}_{P}\left[
        \log \frac{q(Y)}{p_0(Y)}\right] \geq \\ & 
        \inf_{P \in \cH_1} {\mathbb E}_{P}\left[
        \log \frac{q_{\shtarkov,\hat{r}}(Y)}{p_0(Y)}\right] = 
       \inf_{\mu \in \meanspace_1} \inf_{P \in \cH_1: {\mathbb E}_P[Y] = \mu}
        {\mathbb E}_{P}\left[
        \log \frac{\bar{p}_{\hat{r}(Y)}(Y)}{p_0(Y)}\right] - \mmreg_n(\hat{r}) \geq
\\ &   \inf_{\mu \in \meanspace_1} \inf_{P \in \cH_1: {\mathbb E}_P[Y] = \mu} {\mathbb E}_{P}\left[
        \log \frac{\bar{p}_{\mu}(Y)}{p_0(Y)}\right] - \mmreg_n(\hat{r}) \overset{*}{=} \\ &
        \inf_{\mu \in \meanspace_1} {\mathbb E}_{\bar{P}_{\mu}}\left[
        \log \frac{\bar{p}_{\mu}(Y)}{p_0(Y)}\right] - \mmreg_n(\hat{r}) =  \inf_{\mu \in \meanspace_1} D(\bar{P}_{\mu} \| P_0) - \mmreg_n(\hat{r})
= \\ & n \underline{D}_1  - \frac{d-1}{2} \log n + O(1),
        \end{align*}
        where $(*)$ follows by rewriting the quantity inside the logarithm in terms of the mean-value parameterization and evaluating the expectation.
Combining the two displays above, the result follows.
\end{proof}

\section{Proof of Lemma~\ref{lem:brown} and further discussion of Theorem~\ref{thm:main_new}}
\label{sec:proof}
To prove Lemma~\ref{lem:brown}, we  first provide some background on {\em variation diminishing transformations\/} \citep{brown1981variation}.
\begin{definition}\label{def:SignFunction}
Let $\cX = \{x_1, \ldots, x_n \} $ be a finite subset of $\reals$ with $x_1 < x_2< \ldots < x_n$ and let $g: \cX \rightarrow \reals$ be a function, so that  $(g(x_1), \ldots, g(x_n)) \in \mathbb{R}^n$. We let $S^-(g)$ denote the number of sign changes of sequence $g(x_1), \ldots, g(x_n)$ where we ignore zeros; if $g$ is identically $0$ then we set $S^-(g)$ to $0$.  
\end{definition}
\begin{example}
If $g'(x_1,x_2,x_3) = (-1, 0, 1 )$, $g''(x_1,x_2,x_3)  = (-2, 1, 4 )$ and $g'''(x_1,x_2,x_4)  = (-2, 0, 0, 3 )$, then $S^-(g') = S^-(g'')= S^-(g''') =1$.
\end{example}

\begin{definition}
    Now consider arbitrary $\cX \subset \reals$ and let $g: \cX \rightarrow \reals$. For finite $\cV \subset \cX$, say $\cV= \{x_1, \ldots, x_n\}$ with $x_1 < \ldots < x_n$, we let $g_{\cV}= \{g(x_1), \ldots, g(x_n)\}$. We let $S^-(g)$ be the supremum of $S^-(g_{\cV})$ over all finite subsets $\cV$ of $\cX$.
\end{definition}
Intuitively, $S^-(g)$ is the number of times that the function $g(x) $ changes sign as $x \in \cX \subset \reals$ increases. 

\begin{lemma}\label{lemma:Exp_SVR}{\bf \citep[Example 3.1, Proposition 3.1]{brown1981variation}}
Let $P_0$ and $Y$ be as above (\ref{eq:expfam}), where $Y= Y_1$ is 1-dimensional, so $\cY \subseteq \reals$, and consider the 1-dimensional exponential family generated by $P_0$ as in (\ref{eq:expfam}). In the terminology of \cite{brown1981variation}, this family is  {\em SVR}$_{n}(\reals,\Theta)$ and hence   {\em SVR}$_{n}(\cY,\Theta)$ for all $n$. Rather than giving the precise definition of SVR (`strict variation reducing'), we just state the implication of this fact that we need: for any function $g: \cY \rightarrow \reals$ with $\int |g| d \rho > 0$ and $\gamma: \Theta \rightarrow \reals$ with $\gamma(\theta) := \int p_{\theta}(y)g(y) \rho(dy)$, we have: 
$S^{-}(\gamma) \leq S^{-}(g)$. 
\end{lemma}
In words, for any function $g$ as above, the number of sign changes of ${\bf E}_{P_{\theta}}[g(Y)]$ as we vary $\theta$ is bounded by the number of sign changes of $g$ itself on $y\in \cY \subset \reals$. 
Since, in one-dimensional full exponential families, $\mu(\theta)$ is a continuous, strictly increasing function of $\theta$, this also implies that, for any function $g$, the expectation $\gamma(\mu) := {\bf E}_{\bar{P}_{\mu}}[g(Y)]$ also satisfies $S^{-}(\gamma) \leq S^{-}(g)$.

Now 
for any constant $c \in \reals$, 
any $w^{\circ} \in [0,1]$, 
we set $g_c(y) = c+ \log \frac{
(1- w^{\circ}) \bar{p}_{\mu_1^-}(y)+ w^{\circ} \bar{p}_{\mu_1^+}(y)}{p_0(y)}$.
A little calculation of the 
derivatives shows that $g_c(y)$  is strictly convex on $\conv(\cY)$ and not monotonic. Therefore, $g_c(y)$ has exactly one minimum point $y$ and is strictly monotonic on both sides of $y$. 
Thus, $g_c(y)$ as a function of $y \in \conv(\cY)$ changes sign twice; $g_c(y)$'s domain being restricted to $\cY$ (which is not the same as $\conv(\cY)$ in the discrete case), it changes sign at most twice. Thus, for all $c \in \reals$, 
we have $S^-(g_c) \leq 2$. 
Lemma \ref{lemma:Exp_SVR} thus implies that $S^-(\gamma_c(\mu)) \leq 2$ with $\gamma_c(\mu) := f(\mu,w^{\circ})+c$ where $f(\mu,w)$ is defined as in (\ref{eq:brenda}).   Therefore, we know that $g_0(\mu) = f(\mu,w^{\circ})$ as a function of $\mu$ can at most have one minimum point achieved at some $\mu^*$ and if it has such a minimum point, it must be strictly monotonic on both sides of $\mu^*$.

Now $f(0,w^{\circ}) = {\bf E}_{P_0}[g_0(Y)] = - D(\bar{P}_0 \|P_{W_1^{\circ}}) < 0$; but by (\ref{eq:jantje}), which we already showed, $f(\mu^+_1,w^{\circ}) = f(\mu^-_0,w^{\circ}) > 0$. It follows that a $\mu^*$ as mentioned above must exist, and that it lies in between $\mu^-_0$ and $\mu^+_0$; the result follows. 
\paragraph{Why the case $d > 1$ is complicated}
We only managed to prove a general GROW result for surrounding $\cH_1$ for $d=1$. To give the reader an idea where the difficulties lie, we first discuss the case $d=1$ a little more.
One may wonder why even there, we had to resort, via Lemma~\ref{lem:brown}, to the pretty sophisticated theory of variation diminishing transformations. It would seem much simpler to directly calculate the derivative  $(d/ d\mu) f(\mu,w^{\circ})$ and show that, for appropriate choice of $w^{\circ}$, the derivative is $0$ at some $\mu^*$ within the interval $(\mu^-_1,\mu^+_1)$, and negative to the left and positive to the right of $\mu^*$; this would lead to the same conclusion as just stated. Yet the derivative is  given by
\begin{equation}\label{eq:derivative}
\frac{d}{d\mu } f(\mu,w^{\circ}) = \sigma^2_\mu \cdot \left( {\bf E}_{\bar{P}_{\mu}}
\left[  Y \cdot g_0(Y) \right] - 
{\bf E}_{\bar{P}_{\mu}}
\left[  Y \right] \cdot {\bf E}_{\bar{P}_{\mu}}
\left[  g_0(Y) \right]
\right),
\end{equation}
where $\sigma^2_\mu = {\bf E}_{\bar{P}_{\mu}}[Y^2]-({\bf E}_{\bar{P}_{\mu}}[Y])^2$ is the variance of $\bar{P}_{\mu}$.
While (\ref{eq:derivative}) looks `clean', it is not easy to analyze --- for example, it is not a priori clear whether the derivative can be $0$ in only one point. Taking further derivatives does not help either in this respect; for example, the second derivative is not necessarily always-positive.  

Another `straightforward' route to show the result via differentiation might be the following. We fix any prior with finite support, with positive probability $(1-\alpha) w(\mu^+_1)> 0$ on $\{\mu^+_1\}$ and $(1-\alpha) w(\mu^-_1)> 0 $ on  $\{\mu^-_1\}$ and, for $j=1, \ldots, k$, prior $ \alpha w_j$ on $\mu'_j \in \meanspace_1 \setminus \bd(\meanspace_1^{\comp})$, i.e. $\mu'_j$ 
is not on the boundary of  $\meanspace_1$, so that $\sum_{j=1}^k w_j  = w(\mu^+_1) +w(\mu^-_1)=1$, $0 \leq \alpha \leq 1$. We let
$$
q_{\alpha}(Y) := \sum_{j=1}^k w_j  \bar{p}_{\mu_j}(Y) +
w(\mu^+_1) \bar{p}_{\mu^+_1}(Y)  +w(\mu^-_1) \bar{p}_{\mu^-_1}(Y).
$$
Then, if we could show that the KL divergence
\begin{align}\label{eq:deragain}
{\bf E}_{Q_{\alpha}} \left[
\log \frac{q_{\alpha}(Y)}{p_0(Y)}
\right]
\end{align}
were minimized by setting $\alpha =0$, it would follow, by applying Theorem~\ref{thm:ghk}, that $S_{\grow}$ is given by 
$p^*(Y)/p_0(Y)$ where $p^*(Y)$ must be of the form $w(\mu^-_1) \bar{p}_{\mu^-_1} + w(\mu^+_1) \bar{p}_{\mu^+_1}$.

Yet, if we try to show this by differentiating (\ref{eq:deragain}) with respect to $\alpha$, we end up with a similarly hard-to-analyze expression as (\ref{eq:derivative}), and it is again not clear how to proceed. 

These difficulties with showing the result in a straightforward way, by differentiation, only get exacerbated if $d> 1$. So, instead, we might try to extend the above lemma based on variation diminishing transformations to the case $d > 1$. But, literally quoting \cite[Chapter 2]{Brown86}, `results concerning sign changes for multidimensional families appear very weak by comparison to their univariate cousins', and indeed we have not found any existing result in the literature that allows us to extend the above lemma to $d > 1$.

\section{Discussion, Conclusion, Future Work}\label{sec:discuss}
We have shown how GROW e-variables relative to  alternative $\cH_1$ defined in terms of a set of means $\meanspace_1$ relate to a CSC probability bound on an event defined by the same  $\meanspace_1$. We first considered the case of convex $\meanspace_1$;  here our work consisted mostly of reformulating and re-interpreting existing results.  We then considered nonconvex, surrounding $\meanspace_1$. We showed how GROW and the individual-sequence-regret type of {\em relative\/} GROW again relate to a version of the CSC theorem, and we established some additional results  for the case that $\meanspace_1^{\comp}$ is a fixed-radius KL ball for sample size $1$, whereas we let the actual sample size $n$ grow.  
As far as we are aware, our CSC bounds for surrounding $\meanspace_1$ that are KL balls are optimal for this setting.
It is of some interest though to consider the alternative setting in which, at sample size $n$, we consider a KL ball that has a fixed, or very slowly growing, radius when considering distributions on $Y= n^{-1} \sum_{i=1}^n X_i$ rather than on $X$. Thus, instead of (\ref{eq:klball}) we now set, at sample size $n$, 
\begin{equation}\label{eq:klballb}
\meanspace_1^{\comp}  = \{ \mu: D(\bar{P}_{\mu} \| P_0) <  \underline{D}_n \} =  \{ \mu: D(\bar{P}'_{\mu} \| P'_0) <  \underline{D}_n/n \},\end{equation}
where either $\underline{D}_n = \underline{D}$ is constant, or very slowly growing in $n$. First consider the case that it is constant. Then, in terms of a single outcome, the corresponding ball in Euclidean (parameter space) shrinks at rate $1/{\sqrt{n}}$, the familiar scaling when we consider classical parametric testing. Since the boundary $\bd(\meanspace_1^{\comp})$ now changes with $n$, the asymptotics (\ref{eq:dminusone}) we established above are not valid anymore. Therefore, while the CSC Theorem~\ref{thm:nml} is still valid, it may be hard to evaluate the bound (\ref{eq:mlesharpbound}) that it provides.

Now, for the setting (\ref{eq:klballb}), we may also heuristically apply the multivariate Central Limit Theorem (CLT): a second order Taylor approximation of $D(\bar{P}_{\mu} \| P_0)$ in a neighborhood of $\mu=0$ gives that, up to leading order, 
$D(\bar{P}_{\mu} \| P_0)= \mu^{\top} J(0) \mu$, with $J(\mu)$ the Fisher information matrix of $\cE$ in terms of the mean-value parameterization, which is equal to the inverse of the covariance matrix.  The multivariate CLT then immediately gives that, as $n \rightarrow \infty$, we have that  $P_0(Y \in \meanspace_1^{\comp}) \rightarrow A$, where $A$ is the probability that a normally distributed random vector $V$, i.e. $V \sim N(0,I)$, with $I$ the $(d-1) \times (d-1)$ identity matrix, falls in a Euclidean ball of radius $\sqrt{\underline{D}}$. This implies that the bound (\ref{eq:mlesharpbound}) would only remain relevant if for all large $n$, its right-hand side evaluates to a constant smaller than $1$. We currently do not know if this is the case; it is an interesting question for future work.

Now let us consider the 
 scaling (\ref{eq:klballb}) for the 
case that $\underline{D}_n$ is growing at the very slow rate $a (\log(b+ c\log n)$ for suitable $a, b$ and $c$. \citep{kaufmann2021mixture} give an {\em anytime-valid bound} for this case, in which the right-hand side is also a nontrivial constant (i.e.  $< 1$), for all large enough $n$. 
Again, we do not know if we can replicate such bounds with our analyses --- it is  left for future work to determine this, and to further analyze the relation between anytime-valid bounds and the  bounds we derived here, which are related to e-values and hence indirectly related to anytime-validity, but are not anytime-valid themselves.

\bibliography{savi,references,peter}
\end{document}